\documentclass[runningheads]{llncs}
\usepackage[T1]{fontenc}
\usepackage{hyperref}
\usepackage{graphicx, amsmath, nicematrix, longtable}
\usepackage{tikz}
\usepackage{listings, booktabs, cite, xcolor, multirow, caption}

\definecolor{c1}{RGB}{220,235,255}
\definecolor{c2}{RGB}{255,230,210}
\definecolor{c3}{RGB}{120,200,200}
\title{Appendix}
\newcommand{\cbox}[1]{%
  \begingroup
  \setlength{\fboxsep}{1pt}%
  \colorbox{#1}{\strut\hspace{1.2em}}%
  \endgroup
}

\usetikzlibrary{calc, fit, backgrounds}
\newcommand{\indp}{\textrm{Ind}}
\newcommand{\subp}{\textrm{Sub}}
\newcommand{\homp}{\textrm{Hom}}
\begin{document}

\title{Faster Algorithms for Finding Small Induced Patterns in Sparse Host Graphs}
%
%
\author{Priyanshi Agrawal \and Balagopal Komarath}
\titlerunning{Faster Algorithms for Induced Patterns}
\authorrunning{P. Agrawal and B. Komarath}
%
\institute{Indian Institute of Technology, Gandhinagar, India\\
\email{agrawalpriyanshi15@gmail.com}\\
\email{bkomarath@rbgo.in}}
\maketitle              
\begin{abstract}
  We study algorithms to detect induced $G$-subgraphs, for fixed
  patterns $G$, given host graphs with $n$ vertices and $m$ edges as
  input. We show that:
\begin{itemize}
\item There are at least five $5$-vertex connected graphs $G$ (out of
  21) that can be detected in $\widetilde{O}(nm)$-time.
\item There are at least sixty-five $6$-vertex connected graphs $G$
  (out of 112) that can be detected in $\widetilde{O}(m^2)$-time.
\item The graphs $P_7$ and $C_7$ can be detected in $\widetilde{O}(nm^2)$-time.
\end{itemize}

Our main technical tool is a generalized notion of width of tree
decompositions that we call $(p, q)$-width tree decompositions. We use
$(p, q)$-width tree decompositions to derive algorithms with
running-time $\widetilde{O}(n^pm^q)$ for detecting induced
subgraphs. Our upper-bounds are never worse than existing upper-bounds
and are faster when $m = o(n^2)$. Moreover, for some patterns, such as
$C_7$, our algorithms are optimal under standard complexity theoretic
assumptions.

The new technique we use to obtain fast algorithms for fixed pattern
graphs is not scalable to families of pattern graphs. We further
refine this idea by constructing pattern-based polynomials that
exploit the structure of the tree decomposition, and not just its
width, to obtain algorithms for larger classes of graphs. We show:
\begin{itemize}
\item An $\widetilde{O}(m^{k-1})$-time algorithm for detecting induced
  $P_{2k}$ in bipartite graphs.
\item An $\widetilde{O}(nm^{k-1})$-time algorithm for detecting
  induced $C_{2k}$ in bipartite graphs.
\end{itemize}
These algorithms improve the best known upper-bounds for general
graphs.

\keywords{Combinatorial Graph Algorithms \and Induced Subgraph
  Detection \and Graph Pattern Polynomials.}
\end{abstract}

\section{Introduction}
\label{sec:intro}

We consider the computational problem of finding induced copies of
fixed $k$-vertex pattern graphs $G$ in a host graph $G$ with $n$
vertices and $m$ edges given as input. The straight-forward search
algorithm for this problem takes $O(n^k)$-time, which is
polynomial-time since we regard $k$ as a constant. Williams
\cite{WILLIAMS2005357} showed that even minor improvements to this
trivial upper-bound can be used to speed-up algorithms for NP-complete
problems such as $\text{MAX2CNFSAT}$. In particular, they showed that
the $O(n^\omega)$-time algorithm for triangle finding can be used to
derive a $O(2^{\omega n / 3})$-time algorithm for $\text{MAX2CNFSAT}$.

In a seminal paper, Ne{\v s}etřil and Poljak \cite{Neetil1985OnTC}
showed that every $3k$-vertex pattern graph can be found in
$O(n^{\omega k})$-time and asked the question whether there are
patterns where we can do significantly better. To the best of our
knowledge, there is only one pattern family for which we know a
(slightly) better upper-bound than the one given by Ne{\v s}etřil
and Poljak \cite{Neetil1985OnTC}'s $41$-year-old algorithm. Therefore,
research on finding induced copies of fixed size patterns has focused
on (1) improvements for small fixed patterns, mostly on four to nine
vertices and (2) combinatorial algorithms, where we improve upon the
trivial $O(n^k)$-time bound without using fast matrix multiplication
algorithms (See
\cite{DBLP:conf/soda/WilliamsWWY15,DBLP:conf/fsttcs/BlaserKS18,doi:10.1137/140978211}). For
example, Bl{\"{a}}ser, Komarath, and Sreenivasaiah
\cite{DBLP:conf/fsttcs/BlaserKS18} showed that there is a
$O(n^5)$-time algorithm for finding $C_7$ and recently Dalirrooyfard
and Williams \cite{DBLP:conf/focs/DalirrooyfardW22} showed that this
is optimal assuming $K_5$ \footnote{The clique on five vertices}
requires $O(n^5)$-time for combinatorial algorithms.

In practice, many host graphs in which we wish to find patterns are
sparse. Therefore, there has also been research on deriving faster
algorithms taking sparse host graphs as input. A common restriction is
to assume that the host graph is $d$-degenerate for some fixed $d$
\cite{DBLP:conf/focs/0002R21}. For example, we can find
$C_7$ in $O(n^2)$-time in a given $n$-vertex, $d$-degenerate
graph. The $O(.)$ here hides factors that are exponential in
$d$. Since the degeneracy of a graph can be as high as $n$, these
algorithms are not suitable for arbitrary host graphs.

Another line of research to improve the running-time on sparse host
graphs is to consider the running-time as a function on both the
number of vertices and the number of edges in the host graph. For
example, recently, Komarath, Kumar, Mishra and
Sethia,\cite{DBLP:journals/jcss/KomarathKMS26} have shown that $C_6$
has $\widetilde{O}(m^2)$-time \footnote{We use $\widetilde{O}(.)$ to
  denote that we ignore factors logarithmic in the input size.}
algorithms. They obtained this improvement by defining matched
treewidth, a new notion of width of tree decompositions. A key
sub-routine in their paper is an algorithm to construct small
\emph{arithmetic circuits} (Definition~\ref{def:arithmetic_circuit})
for \emph{graph homomorphism polynomials}
(Definition~\ref{def:homomorphism_polynomial}) when the pattern graph
has small matched treewidth. To use these arithmetic circuits to find
induced patterns, they then construct specific tree decompositions
with the required matched treewidth and an additional property that
the tree decomposition can be used to avoid duplicate homomorphisms to
the same occurence of the pattern as a subgraph while constructing the
circuit. However, their work is mainly about counting (not necessarily
induced) subgraphs.

\paragraph{Our results:} In this paper, we define a new notion of
width for tree decompositions called $(p, q)$-width. The key
ingredient behind our improved algorithms is a
$\widetilde{O}(n^pm^q)$-time algorithm to construct arithmetic
circuits for graph homomorphism polynomials from pattern graphs of
$(p, q)$-width. Then, we construct specific tree decompositions with
certain additional conditions that we call \emph{automorphism avoiding
  conditions} for many pattern graphs on $5$ to $7$ vertices and
derive the following results:

\begin{itemize}
\item There are at least five $5$-vertex connected graphs $G$ (out of
  21) that can be detected in $\widetilde{O}(nm)$-time.
\item There are at least sixty-five $6$-vertex connected graphs $G$
  (out of 112) that can be detected in $\widetilde{O}(m^2)$-time.
\item The graphs $P_7$ and $C_7$ can be detected in
  $\widetilde{O}(nm^2)$-time.
\end{itemize}

We derive $O(nm)$, $O(m^2)$, and $O(nm^2)$ algorithms for detecting
induced $P_5$, $P_6$, and $P_7$ respectively. To the best of our
knowledge, the previous best algorithmic upper bounds which are given
by Ho{\`{a}}ng, Kaminski, Sawada, and
Sritharan\cite{DBLP:journals/dam/HoangKSS13} for detecting induced
$P_5$, $P_6$, and $P_7$ are $O(m^2)$, $O(nm^2)$, and $O(m^3)$
respectively when considering both number of vertices and edges of the
host graph. Therefore, the above results improve existing bounds for
lengths $5$, $6$, and $7$ for the cases where
$m = \Theta(n^{1+\epsilon})$, $0 \leq \epsilon < 1$.

We give a $\widetilde{O}(nm)$-time algorithm for $C_5$ which is
optimal. This algorithm is an improvement on the $O(n^3)$-time
algorithm given by Komarath, Kumar, Mishra, and Sethia
\cite{DBLP:journals/jcss/KomarathKMS26} when $m = o(n^2)$. We stress
that simply having a $(1, 1)$-width tree decomposition for a pattern
is not sufficient to derive a faster algorithm for finding induced
copies. For example, the complete bipartite graph $K_{2,3}$ has a
$(1, 1)$-width tree decomposition. But, we could not derive
$\widetilde{O}(nm)$-time algorithms for finding it as those tree
decompositions do not have automorphism avoiding conditions. Our
techniques for finding induced subgraphs are illustrated using an
example in Section~\ref{sec:p4inbull}.

One drawback of our method (and many existing methods such as the ones
in
\cite{doi:10.1137/140978211,DBLP:conf/soda/WilliamsWWY15,DBLP:conf/fsttcs/BlaserKS18}
used to derive faster algorithms for finding induced patterns) is that
the analysis of specific subgraphs do not scale to larger
patterns. For example, our algorithm for $C_7$ requires constructing
specific tree decompositions for over 100 graphs. Therefore, we do not
believe it will be possible to extend our analysis, or similar
approaches, to derive algorithms for $C_k$ for arbitrary $k$. However,
we show that by imposing natural restrictions on the host graph, for
example, restricting them to bipartite graphs, the required analysis
may become simple enough for infinite size pattern families. We show:

\begin{itemize}
\item An $\widetilde{O}(m^{k-1})$-time algorithm for detecting induced
  $P_{2k}$ in bipartite graphs.
\item An $\widetilde{O}(nm^{k-1})$-time algorithm for detecting
  induced $C_{2k}$ in bipartite graphs.
\end{itemize}

To the best of our knowledge, for general host graphs, the fastest
algorithms to find induced $P_{2k}$ runs in time $O(n^{2k-2})$
\cite{DBLP:conf/fsttcs/BlaserKS18}. Note that our algorithm is never
worse than this and is significantly faster when $m = o(n^2)$, with
the catch that ours only work when the host graph is
bipartite. Similarly, our algorithm for $C_{2k}$ improves upon the
trivial $O(m^k)$-time algorithm when $m = \omega(n)$.

Observe that an $\widetilde{O}(m^k)$-time algorithm is trivial for
even vertex paths and cycles in general graphs. By assuming that the
host is bipartite, we slighly improve on this running-time, and only
for $m = \Theta(n^{1+\epsilon})$, $0 \leq \epsilon < 1$ for cycles. We
also note that Ho{\`{a}}ng, Kaminski, Sawada and
Sritharan\cite{DBLP:journals/dam/HoangKSS13} gave a
$\widetilde{O}(nm^{k-1})$-time algorithm for finding $P_{2k}$ in
\emph{general graphs} (They do not give algorithms for finding induced
cycles).

\section{Preliminaries}

We begin by defining standard notions of containment in graph theory.

\begin{definition}[Subgraph]
\label{def:subgraph}
Let $H=(V_H,E_H)$ and $G=(V_G,E_G)$ be graphs. We say that
$G$ is a \emph{subgraph} of $H$, denoted by $G\subseteq H$, if
\[
V_G\subseteq V_H
\quad\text{and}\quad
E_G\subseteq
\left\{\{u,v\}\in E_H : u,v\in V_G\right\}.
\]

\end{definition}

\begin{definition}[Induced Subgraph]
    \label{def:induced-subgraph}
Let $G = (V, E)$ be a graph and let $T \subseteq V$.  
The induced subgraph of $G$ on $T$, denoted $G[T]$, is the graph whose vertex set is $T$ and whose edge set consists of all edges of $G$ with both endpoints in $T$.
\end{definition}

\begin{definition}[Minor]
    \label{def:minor}
    Let $H$ and $G$ be graphs. We say that $G$ is a
    \emph{minor} of $H$ (denoted $G \preceq H$) if $G$ can be obtained
    from $H$ by a finite sequence of the operations like vertex
    deletions, edge deletions and edge contractions.
\end{definition}

\begin{definition}[Edge Contraction]
\label{def:edge-contraction}
Let $H=(V_H,E_H)$ be a graph and let $e=\{u,v\}\in E_H$.
The graph obtained by \emph{contracting $e$}, denoted by $H/e$, is
defined as follows.

Let
\[
E'=\{\ell\in E_H : u\in \ell \text{ or } v\in \ell\}
\]
be the set of edges incident to $u$ or $v$, and let
\[
E''=
\bigl\{\{u',w\} : w\in
(N_H(u)\cup N_H(v))\setminus\{u,v\}\bigr\}.
\]

where
\begin{align*}
H/e&=(V_{H/e},E_{H/e})\\
V_{H/e}&=(V_H\setminus\{u,v\})\cup\{u'\}\\
E_{H/e}&=(E_H\setminus E')\cup E''
\end{align*}.

Here, $u'$ is a new vertex representing the contraction of
$u$ and $v$.
\end{definition}

We now define tree decompositions and path decompositions, which are useful for designing dynamic programming algorithms on graphs.

\begin{definition}[Tree Decomposition]
  \label{def:tree-decomposition}
  A \emph{tree decomposition} of a graph $G=(V,E)$ is a pair
  $(\mathcal{T}, \{X_t\}_{t \in V(\mathcal{T})})$ where $\mathcal{T}$
  is a tree and each node $t$ is assigned a subset $X_t \subseteq V$,
  called a bag, such that: (1)
  $\bigcup_{t \in V(\mathcal{T})} X_t = V$, (2) for every edge
  $\{u,v\} \in E$, there exists $t$ with $\{u,v\} \subseteq X_t$, and
  (3) for every vertex $v \in V$, the set $\{t \mid v \in X_t\}$
  induces a connected subtree of $\mathcal{T}$.  The width of a tree
  decomposition is $\max_{t} |X_t| - 1$.  The treewidth of $G$,
  denoted $\mathrm{tw}(G)$, is the minimum width over all tree
  decompositions of $G$.
\end{definition}

\begin{definition}[Path Decomposition]
    \label{def:path-decomposition}
A path decomposition is a tree decomposition in which the underlying tree $\mathcal{T}$ is a path.  
The pathwidth of $G$, denoted $\mathrm{pw}(G)$, is the minimum width over all path decompositions of $G$.
\end{definition}

Finally, we define arithmetic circuits.

\begin{definition}[Arithmetic Circuit]
    \label{def:arithmetic_circuit}
An arithmetic circuit over a ring $R$ is a directed acyclic graph whose nodes are of three types:
(1) input nodes labeled by variables or constants from $R$,
(2) addition gates labeled $+$,
(3) multiplication gates labeled $\times$.

Each gate computes a polynomial over $R$ in the natural way.  
The output gate computes the polynomial defined by the circuit.  

The \emph{size of a circuit} is the number of gates, and the \emph{depth} is the length of the longest directed path from an input node to the output gate.
\end{definition}

\begin{definition}[Induced Subgraph Polynomial]
    \label{def:induced-subgraph-polynomial}
For a $k$-vertex pattern $G$ and a host $H$, we define the
\emph{induced subgraph polynomial} as
\begin{align*}
\indp_{G,H}\bigl((y_v)_{v\in V(H)}\bigr)
&=
\sum_{\substack{S\subseteq V(H)\\ |S|=k\\ H[S]\cong G}}
\prod_{v\in S} y_v.
\end{align*}
\end{definition}

\begin{definition}[Subgraph Polynomial]
    \label{def:subgraph-polynomial}
For a $k$-vertex pattern $G$ and a host $H$, we define the
\emph{subgraph polynomial} as
\begin{align*}
\subp_{G,H}\bigl((y_v)_{v\in V(H)}\bigr)
&=
\sum_{\substack{S\subseteq V(H)\\ |S|=k\\ H[S]\supseteq G}}
\alpha_S\prod_{v\in S} y_v,
\end{align*}
where $H[S]\supseteq G$ means that $H[S]$ contains a subgraph
isomorphic to $G$, and $\alpha_S$ denotes the number of such
subgraphs in $H[S]$.
\end{definition}

\begin{definition}[Homomorphism Polynomial]
    \label{def:homomorphism_polynomial}
For a $k$-vertex pattern $G$ and a host $H$, we define the
\emph{homomorphism polynomial} as
\begin{align*}
\homp_{G,H}\bigl((y_v)_{v\in V(H)}\bigr)
&=
\sum_{\phi}
\prod_{v\in V(G)} y_{\phi(v)},
\end{align*}
where $\phi$ ranges over all homomorphisms from $G$ to $H$. A
\emph{homomorphism} is a map $\phi:V(G)\to V(H)$ such that
\[
\{u,v\}\in E(G)
\implies
\{\phi(u),\phi(v)\}\in E(H).
\]
\end{definition}

\section{Small Circuits for Homomorphism Polynomials from Tree
  Decompositions}
\label{sec:gtd}

We begin by defining a new notion of width for tree decompositions.

\begin{definition}
\label{def:pqwidth}
A tree decomposition $T$ of $G$ has \emph{width $(p, q)$} if for each
bag $B$ of $T$, there exist vertices $u_1, \dotsc, u_i \in B$ with
$i \leq p$ and edges $e_1, \dotsc, e_j \in E(G[B])$ with $j \leq q$
such that
    \[
        E\bigl(G[B] \setminus \{u_1, \dotsc, u_i\}\bigr) = \{e_1, \dotsc, e_j\},
    \]
    and moreover every vertex of $B \setminus \{u_1,\dotsc,u_i\}$ is an endpoint
    of some $e_\ell$, $\ell \le j$.

    We say that the vertices $u_1, \dotsc, u_i$ and edges
    $e_1, \dotsc, e_j$ is a \emph{$(p, q)$-witness} of the bag $B$.
  \end{definition}

  The same edge could occur in multiple bags in a tree
  decomposition. We have to ensure that an edge occuring in the image
  of a homomorphism is only multiplied in once in the term
  corresponding to the homomorphism. We define a notion of \emph{local vertex} for this purpose.
\begin{definition}
\label{def:local-edge}
Let $T$ be a tree decomposition of $G$. Fix a function $f: V(G) \to V(T)$
assigning each vertex $v \in V(G)$ to some bag $f(v)$ (such a bag exists by the definition of tree
decomposition). We say $v$ is \emph{local} to bag $B$ if $f(v) = B$, and
write
\[
    V_B = \{ v \in V(G) : f(v) = B \}.
\]
\end{definition}
Note that since $f$ is a function, each vertex of $G$ is local to
exactly one bag in the tree decomposition.

We now prove that we can efficiently construct circuits for
homomorphism polynomials from $(p, q)$-width tree decompositions. Our
proof is a refinement of the dynamic programming algorithm by
D{\'{\i}}az, Serna, and Thilikos \cite{DBLP:journals/tcs/DiazST02}.
\begin{theorem}
\label{thm:dp-pq}
Let $G$ be a $k$-vertex pattern graph having a tree decomposition $T$
of $(p,q)$-width, and let $H$ be a host graph with $n$ vertices and
$m$ edges. Then an arithmetic circuit of size $O(n^p m^q)$ computing
$\homp_{G,H}$ can be constructed in time $\widetilde{O}(n^p m^q)$.
\end{theorem}
\begin{proof}
We describe the construction of the arithmetic circuit for $\homp_{G,H}$.

Let $T$ be a tree decomposition of $G$ of $(p,q)$-width, rooted at
the bag $R$. For any bag $B$ in $T$, let $T_B$ denote the subtree of
$T$ rooted at $B$, and define

$$G_B = G[V'], \quad \text{where } V' = \bigcup_{B'\in V(T_B)} B'$$.

The circuit will have two kinds of gates, denoted by $g_{B,\phi}$ and
$h_{B,\phi}$, where $B$ is a bag in $T$ and $\phi$ is some partial
homomorphism from $G$ to $H$ defined on a subgraph of $G$ to be
specified later.

We define $g_{B, \phi}$ so that the polynomial computed at
$g_{B,\phi}$ is:

\begin{equation}\label{eq:gbphi}
g_{B,\phi} = \sum_{\substack{
\phi':G_B\rightarrow H\\
\phi'\sim \phi
}}
\;
\prod_{\substack{
v\in V(G_B)\\
v\text{ is local to some bag in }T_B
}}
y_{\phi'(v)}
\end{equation}

Here, $\phi$ is any homomorphism from $G[B]$ to $H$. We say that
$\phi'$ is \emph{compatible} with $\phi$, denoted by
$\phi' \sim \phi$, iff
\[
\phi'(v)=\phi(v) \quad \text{for every } v \text{ where $\phi(v)$ and $\phi'(v)$ are defined}.
\]

Let $B$ be any non-root bag with a parent $B'$. We add gates
$h_{B,\phi}$, where $\phi$ is a homomorphism from $G[B \cap B']$ to
$H$. These gates will compute the following polynomial:

\begin{equation}\label{eq:hbphi}
h_{B,\phi}
=
\sum_{\substack{\phi' : G[B] \mapsto H\\ \phi'\sim\phi}} g_{B,\phi'}.
\end{equation}

We now construct the circuit by traversing the tree decomposition $T$
in a bottom-up fashion.

\paragraph{Leaf bags:}
Suppose that $B$ is a leaf of $T$. Then $G_B=G[B]$. There are no child
bags, and hence all vertices of $G_B$ that are local to $B$ can be
introduced directly. For every $\phi : G[B] \mapsto G$, we create the
product gate

$$
g_{B,\phi}
=
\prod_{\substack{
v\in V(G[B])\\
v\text{ is local to }B
}}
y_{\phi(v)}.
$$

Note that we can enumerate all such $\phi$ in $O(n^p m^q)$ time by
enumerating all possible images of the $(p, q)$-witness
$u_1,\dotsc,u_i$ and $e_1,\dotsc,e_j$ of the bag $B$ where $i\leq p$
and $j \leq q$. This can be done by enumerating all elements of
$V(H)^i \times E(H)^j$.

\paragraph{Non-leaf bags:}
Now suppose that $B$ is a non-leaf bag with children
$B_1,\ldots,B_r$. For every homomorphism $\phi : G[B] \mapsto H$, we
add the following gate:

\begin{equation}
  \label{eq:gbphi-gate-nonleaf}
g_{B,\phi}
=
\left(
\prod_{i=1}^{r}
h_{B_i,\phi|_{B\cap B_i}}
\right)
\cdot
\left(
\prod_{\substack{
v\in V(G[B])\\
v\text{ is local to }B
}}
y_{\phi(v)}
\right).
\end{equation}

\paragraph{The gates $h_{B,\phi}$:} The gates $h_{B,\phi}$ can be
constructed as follows: When we construct a gate $g_{B,\phi'}$ for
some $\phi' : G[B] \mapsto H$, we compute $\phi \sim \phi'$ such that
$\phi : G[B \cap B'] \mapsto H$ and then add the gate $g_{B,\phi'}$ as
an input to the $+$-gate $h_{B,\phi}$.

\paragraph{The output gate:} Finally, the output gate of our circuit
is a $+$-gate that adds all $g_{R,\phi}$, where $R$ is the root bag in
$T$.

\paragraph{Correctness:} We prove that the gates $g_{B,\phi}$ compute
the polynomial in Equation~\ref{eq:gbphi} using induction. Note that
$G_R = G$ and therefore our output gate will correctly compute
$\homp_{G,H}$ given that all $g_{B,\phi}$ and $h_{B,\phi}$ are as
claimed.

In the base case, $B$ is a leaf bag and the construction directly
computes the required polynomial.

In the inductive case, $B$ is a non-leaf bag with children
$B_1,\dotsc,B_r$. Let $\phi : G_B \mapsto H$. First, we show that the
the monomial corresponding to $\phi$ in Equation~\ref{eq:gbphi} is
generated by the gate $g_{B,\phi'}$ as in
Equation~\ref{eq:gbphi-gate-nonleaf} where $\phi' : G[B] \mapsto H$ is
the unique homomorphism such that $\phi' \sim \phi$. Similarly, let
$\phi_i : G_{B_i} \mapsto H$ be the unique homomorphisms such that
$\phi_i \sim \phi$ for $1 \leq i \leq r$. Let
$\phi_i' : G[B \cap B_i] \mapsto H$ be the unique homomorphisms such
that $\phi_i' \sim \phi_i$ for $1 \leq i \leq r$. By the induction
hypothesis, $h_{B_i,\phi_i'}$ will contain the monomial corresponding
to $\phi_i$ (notice that $\phi |_{B \cap B_i} = \phi_i'$). For this
choice of $\phi_i'$ for all $1 \leq i \leq r$, the first factor in the
product in Equation~\ref{eq:gbphi-gate-nonleaf} will include all
$y_{\phi(v)}$ where $v$ is local to some bag in $T_B$ that is not
$B$. The second factor will include all $y_{\phi'(v) (= \phi(v))}$ for
$v$ that is local to $B$. Therefore the monomial corresponding to
$\phi' = \phi$ in the summation of Equation~\ref{eq:gbphi} is
generated at the gate Equation~\ref{eq:gbphi-gate-nonleaf}. Observe
that since the decomposition of $\phi$ into $\phi'$ and $\phi_i$ for
$1 \leq i \leq r$ in the above argument is unique, the monomial
corresponding to $\phi$ is also generated exactly once.

We now show that all monomials generated by the gate $g_{B,\phi'}$ as
in Equation~\ref{eq:gbphi-gate-nonleaf} correspond to valid
homomorphisms. Observe that inductively the contribution of the first
factor in all monomials generated correspond to valid homomorphisms
$\phi_i : G_{B_i} \mapsto H$. By the connectivity property of tree
decompositions, any vertex $v$ appearing in both $G_{B_i}$ and
$G_{B_j}$ for $i \neq j$ must also appear in $B$. Therefore,
$\phi_i(v) = \phi_j(v) = \phi'(v)$ ensuring that any monomial produced
at this gate corresponds to a valid homomorphism.

The gates $h_{B,\phi}$ are correct because we compute them in exactly
the same way we defined them in Equation~\ref{eq:hbphi}.

\paragraph{Time complexity of the construction:} We now prove that the
running-time of the construction algorithm is
$\widetilde{O}(n^pm^q)$. Since the pattern graph $G$ is fixed, there
are at most $O(n^pm^q)$ pairs $(B, \phi)$ where
$\phi : G[B] \mapsto H$. The gates $g_{B,\phi}$ are constructed by
simply iterating over all such pairs and the amount of work done per
gate is logarithmic in the size of the input. This is because give
$(B, \phi)$, we can compute the indices of the $h$ gates in
constant-time and then look up $h_{B_i,\phi|_{B\cap B_i}}$ in
$O(\log n)$-time. Notice that the size of the circuit is $O(n^p m^q)$
as each wire in the circuit is incident on some $g_{B,\phi}$ which has
only constant number of inputs and feeds its output to exactly one $h$
gate.
\end{proof}

\section{Using Homomorphisms for Finding Induced Subgraphs}
\label{sec:overview}

Our algorithms are generalizations of algorithms for detecting small
induced patterns given by Bl\"aser, Komarath, and Sreenivasaiah
\cite{DBLP:conf/fsttcs/BlaserKS18} and Komarath, Kumar, Mishra, and
Sethia \cite{DBLP:journals/jcss/KomarathKMS26}. For completeness, we
describe the general method here.

Using inclusion-exclusion, the induced subgraph polynomial
(Definition~\ref{def:induced-subgraph-polynomial}) can be expressed in
terms of subgraph polynomials
(Definition~\ref{def:subgraph-polynomial}) as
\[
\indp_{G,H}
=
\sum_{G'\supseteq G}
(-1)^{|E(G')|-|E(G)|}
\alpha_{G,G'}\subp_{G',H},
\]
where $\alpha_{G,G'}$ denotes the number of subgraphs of $G'$ that
are isomorphic to $G$.

\begin{example}
  Let $H$ be the graph $K_3+e$ which is a triangle on vertices
  $1$,$2$, and $3$ and an edge $\{3,4\}$. Then, we have
  $\indp_{P_3, H} = y_1y_3y_4 + y_2y_3y_4$,
  $\subp_{P_3,H} = 3y_1y_2y_3 + y_1y_3y_4 + y_2y_3y_4$, and
  $\subp_{K_3, H} = y_1y_2y_3$. By the above inclusion-exclusion
  formula, we have $\indp_{P_3, H} = \subp_{P_3,H}-3\subp_{K_3,H}$.
\end{example}

Given an $H$, the problem of detecting whether there is an induced $G$
in $H$ is the same as checking whether the polynomial $\indp_{G,H}$ is
not the identity (zero) polynomial. Since the coefficients of all
monomials are either $0$ or $1$ in $\indp_{G, H}$, this is equivalent
to testing whether it is non-zero modulo $p$ for any prime $p$. In the
above example, we can take $p = 3$ and then checking whether
$\indp_{P_3, H} \neq 0$ is equivalent to checking
$\subp_{P_3, H} \neq 0$. That is, taking modulo $3$ eliminated the
$\subp_{K_3, H}$. It remains to compute $\subp_{P_3, H}$ efficiently
given $H$. However, we don't know how to do this. Instead, we compute
a polynomial where the multilinear monomials correspond to the
monomials of $\subp_{P_3, H}$ and then test whether this new
polynomial has a multilinear monomial.

The multilinear monomials of $\homp_{P_3,H}$ are in $2$-to-$1$
correspondence with the monomials of $\subp_{P_3,H}$. This is because
$P_3$ has two automorphisms. Therefore, we have now reduced the
problem to that of checking whether $\homp_{P_3,H}$ has a multilinear
monomial modulo $3$. There is a linear-time algorithm to check whether
the polynomial has a multilinear monomial modulo $p$ given a circuit
computing it for any $p$ given by Koutis and Williams\cite{Koutis2009LIMITSAA}. Moreover, the polynomial $\homp_{P_3, H}$
can be computed in $O(n+m)$-time as
$\sum_{v\in V(H)} y_v(\sum_{\{u, v\}\in E(H)}y_u)^2$. Therefore, we
have a $O(n+m)$-time algorithm for detecting induced $P_3$.

For the pattern $P_4$, the above method is not as straightforward as
we have to do arithmetic modulo $2$ to obtain
$\indp_{P_4, H} = \subp_{P_4, H}$. But $P_4$ also has two
automorphisms. Therefore, modulo $2$ arithmetic will eliminate all
multilinear monomials from $\homp_{P_4,H}$. We can circumvent this
problem by ensuring that only the multilinear monomials corresponding
to $P_4$-homomorphisms where $\phi(2) < \phi(3)$ survive. We call this
condition an \emph{automorphism avoiding condition}. Now, the
multilinear monomials of that polynomial will be the same as that of
$\subp_{P_4,H}$. We sketch the part of the dynamic programming
algorithm based on a tree decomposition of $P_4$ to show how this is
ensured. Consider $P_4$ labeled as
$1 \text{---} 2 \text{---} 3 \text{---} 4$ and the tree decomposition
$B_1 \text{---} B_2 \text{---} B_3$ where $B_1 = \{1, 2\}$,
$B_2 = \{2, 3\}$ and $B_3 = \{3, 4\}$. For each bag $B_i$, the
algorithm iterates over elements of $V(H)^{|B_i|}$ and checks whether
we can map the vertices in $B_i$ to the elements in this tuple. So
when we iterate over $V(H)^2$ for the bag $B_2$ and we consider the
element $(u, v)$, we will check that $\{u, v\} \in E(H)$ and that
$u < v$ (This is the check $\phi(2) < \phi(3)$). Therefore for any
$P_4$-subgraph where $\{u, v\}$ is the middle edge, we only account
for it once in the final polynomial.

We have already seen in Section~\ref{sec:gtd} how to use
$(p, q)$-width tree decompositions to produce circuits of size
$O(n^pm^q)$ for homomorphism polynomials. In the coming sections, we
show that these improved constructions along with the above idea can
be used to derive faster algorithms for many patterns.

\section{Faster Algorithms for Small Patterns}
\label{sec:small}
  
For detecting induced isomorphisms, simply constructing the
homomorphism polynomial using Theorem~\ref{thm:dp-pq} is not
sufficient. As described in Section~\ref{sec:overview}, we also have to show
existence of decompositions that avoids automorphisms. In this
section, we do that for many small pattern graphs. Note that since
$(p, q)$-width tree decompositions are tree decompositions
$p + 2q \geq w + 1$ where $w$ is the treewidth. So for a non-tree,
five vertex pattern, the best we can hope for is $p = 1, q = 1$. We
show that there are five $5$-vertex patterns that have $(1, 1)$-width
tree decompositions and can be detected in $\widetilde{O}(nm)$-time.

\begin{theorem}\label{thm:g5}
  Given an $n$-vertex, $m$-edge host graph $H$ as input, there are
  $\widetilde{O}(nm)$-time algorithms to detect induced subgraphs
  for the five $5$-vertex graphs G05, G09, G12 to G14 in
  Table~\ref{tab:g5} and in Figure~\ref{fig:selected-five-vertex-graphs}.
\end{theorem}
\begin{proof}
  By the method in Section~\ref{sec:overview}, to prove that a $5$-vertex
  pattern $G$ has an $\widetilde{O}(nm)$-time algorithm, it is
  sufficient to provide a $(1, 1)$-decomposition and the automorphism
  avoiding conditions for that tree decomposition for all $5$-vertex
  patterns $G'$ that contain $G$ as a subgraph an odd number of times.

  Table~\ref{tab:g5} lists all $5$-vertex graphs by specifying their
  adjacency matrix(in the \emph{Adjacency} column) as follows: Assume the vertices are labeled
  $0,\dotsc,4$, then the first bit is for the edge $\{0,1\}$, the
  second for $\{0,2\}$ and so on until the $\{3,4\}$ edge. The number of automorphisms is given in the \emph{\#Aut} column. All tree
  decompositions are path decompositions. So we specify the path of
  bags in the \emph{t.d.} column. The \emph{Checks} column specifies conditions that have to be
  checked failing which we should eliminate the corresponding
  monomials from the homomorphism polynomial.The - represents that the graph occurs an odd number of times in the $K_4$ minor graph and hence cannot be detected using this method. The $\sim$ represents that the even number of occurences in the graph cannot be avoided by using any checks from the bags. Finally, the last column
  \emph{\#AutRem} specifies the number of automorphisms that will remain
  after these checks are applied.

  The $G_i \times G_j$ entry of Table~\ref{tab:subg5} specifies the
  number of $G_i$-subgraphs in $G_j$. It contains all 21 connected
  $5$-vertex patterns. The graphs G16 to G21 have a $K_4$-minor
  (marked by a ${}^{K_4}$ superscript) and therefore cannot have a
  $(1, 1)$-width tree decomposition. The graphs G07 and G08 have
  $(1, 1)$-width tree decompositions but there are no checks (marked
  by ${}^?$ superscript) that can be done to make the number of
  remaining automorphisms odd. Now, for any graph $G$ that is
  contained in one of G07, G08, G16--21 an odd number of times, we do
  not have $\widetilde{O}(nm)$-time algorithms (marked by ${}^*$
  superscript). For the remaining 5 graphs G05, G09, G12--14 (See
  Figure~\ref{fig:selected-five-vertex-graphs}), we can use the
  $(1,1)$-width tree decompositions and the checks in
  Table~\ref{tab:g5} to obtain $\widetilde{O}(nm)$-time algorithms.

\begin{longtable}{ccrllr}
\textbf{Name} & \textbf{Adjacency} & \textbf{\#Aut} &
\textbf{t.d.} & \textbf{Checks} & \textbf{\#AutRem} \\
\hline
\endfirsthead
\textbf{Name} & \textbf{Adjacency} & \textbf{\#Aut} &
\textbf{t.d.} & \textbf{Checks} & \textbf{\#AutRem} \\
\hline
\endhead
G01 & 0001001011 & 24 &124-024-034  & - & - \\
G02 & 0011001010 &2 &034-142   & $1<2$ & 1 \\
G03 & 0011001011 & 4&034-124-234  & $1<2,2<3,3<0$ & 1 \\
G04 & 0011011010 & 2&013-014-24   & $0<1$ & 1\\
G05 & 0011010011 & 2&013-034-24   & $3<4$ & 1 \\
G06 & 0011011011 &2 &034-134-24   & $0<1$ & 1 \\
G07 & 0011011110 & 12 &034-134-234  & $\sim$ & - \\
G08 & 0011011111 & 12 &034-134-234  & $\sim$ & - \\
G09 & 0101011000 & 2&042-041-13 & $0<1$ & 1 \\
G10 & 0101011010 & 2&042-134  & $0<2$ & 1 \\
G11 & 0101011011 & 8&024-034-134  & $2<0,0<3,3<1$ & 1 \\
G12 & 0110011010 & 10 &143-034-024 & $3<4$ & 5 \\
G13 & 0111011010 & 2 &024-014-013 & $0<4$ & 1 \\
G14 & 0111011011 & 2 &134-034-024 & $0<3$ & 1 \\
G15 & 0111001110 & 2 &14-024-023 & $0<2$ & 1 \\
G16 & 0111001111 & 6 &- & - & - \\
G17 & 0111011111 & 4 &- & - & - \\
G18 & 0111111010 & 4 &- & - & - \\
G19 & 0111111011 & 8 &- & - & - \\
G20 & 0111111111 & 12 &- & - & - \\
G21 & 1111111111 & 120 &- & - & - \\
\caption{$(1, 1)$-width tree decompositions for $5$-vertex graphs}

\label{tab:g5}
\end{longtable}

\begin{figure}[ht]
\centering
\scriptsize

\tikzset{
    vertex/.style={
        circle,
        fill=black,
        inner sep=2.3pt
    },
    edge/.style={
        black,
        thin
    },
    bag/.style={
        draw,
        minimum width=0.72cm,
        minimum height=0.32cm,
        inner sep=1.5pt,
        font=\tiny
    }
}


\begin{minipage}[t]{0.29\textwidth}
\centering
\textbf{G05}

\begin{tikzpicture}[scale=0.55, every node/.style={font=\tiny}]

\node[vertex,label=above:0] (v0) at (0,1.4) {};
\node[vertex,label=right:1] (v1) at (1.2,0.5) {};
\node[vertex,label=right:2] (v2) at (0.75,-0.9) {};
\node[vertex,label=left:3] (v3) at (-0.75,-0.9) {};
\node[vertex,label=left:4] (v4) at (-1.2,0.5) {};

\draw[edge] (v0)--(v3);
\draw[edge] (v0)--(v4);
\draw[edge] (v1)--(v3);
\draw[edge] (v2)--(v4);
\draw[edge] (v3)--(v4);

\node[bag] (b1) at (2.5,0.9) {013};
\node[bag] (b2) at (2.5,0.0) {034};
\node[bag] (b3) at (2.5,-0.9) {24};

\draw (b1.south)--(b2.north);
\draw (b2.south)--(b3.north);

\end{tikzpicture}

\vspace{0.5mm}

$\sigma(3)<\sigma(4)$
\end{minipage}
\hfill
\begin{minipage}[t]{0.29\textwidth}
\centering
\textbf{G09}

\begin{tikzpicture}[scale=0.55, every node/.style={font=\tiny}]

\node[vertex,label=above:0] (v0) at (0,1.4) {};
\node[vertex,label=right:1] (v1) at (1.2,0.5) {};
\node[vertex,label=right:2] (v2) at (0.75,-0.9) {};
\node[vertex,label=left:3] (v3) at (-0.75,-0.9) {};
\node[vertex,label=left:4] (v4) at (-1.2,0.5) {};

\draw[edge] (v0)--(v2);
\draw[edge] (v0)--(v4);
\draw[edge] (v1)--(v3);
\draw[edge] (v1)--(v4);

\node[bag] (b1) at (2.5,0.45) {042};
\node[bag] (b2) at (2.5,-0.45) {041};
\node[bag] (b3) at (2.5,-1.35) {13};

\draw (b1.south)--(b2.north);
\draw (b2.south)--(b3.north);

\end{tikzpicture}

\vspace{0.5mm}

$\sigma(0)<\sigma(1)$
\end{minipage}
\hfill
\begin{minipage}[t]{0.29\textwidth}
\centering
\textbf{G12}

\begin{tikzpicture}[scale=0.55, every node/.style={font=\tiny}]

\node[vertex,label=above:0] (v0) at (0,1.4) {};
\node[vertex,label=right:1] (v1) at (1.2,0.5) {};
\node[vertex,label=right:2] (v2) at (0.75,-0.9) {};
\node[vertex,label=left:3] (v3) at (-0.75,-0.9) {};
\node[vertex,label=left:4] (v4) at (-1.2,0.5) {};

\draw[edge] (v0)--(v2);
\draw[edge] (v0)--(v3);
\draw[edge] (v1)--(v3);
\draw[edge] (v1)--(v4);
\draw[edge] (v2)--(v4);

\node[bag] (b1) at (2.5,0.9) {143};
\node[bag] (b2) at (2.5,0.0) {034};
\node[bag] (b3) at (2.5,-0.9) {024};

\draw (b1.south)--(b2.north);
\draw (b2.south)--(b3.north);

\end{tikzpicture}

\vspace{0.5mm}

$\sigma(3)<\sigma(4)$
\end{minipage}

\vspace{2mm}


\begin{minipage}[t]{0.29\textwidth}
\centering
\textbf{G13}

\begin{tikzpicture}[scale=0.55, every node/.style={font=\tiny}]

\node[vertex,label=above:0] (v0) at (0,1.4) {};
\node[vertex,label=right:1] (v1) at (1.2,0.5) {};
\node[vertex,label=right:2] (v2) at (0.75,-0.9) {};
\node[vertex,label=left:3] (v3) at (-0.75,-0.9) {};
\node[vertex,label=left:4] (v4) at (-1.2,0.5) {};

\draw[edge] (v0)--(v2);
\draw[edge] (v0)--(v3);
\draw[edge] (v0)--(v4);
\draw[edge] (v1)--(v3);
\draw[edge] (v1)--(v4);
\draw[edge] (v2)--(v4);

\node[bag] (b1) at (2.5,0.9) {024};
\node[bag] (b2) at (2.5,0.0) {014};
\node[bag] (b3) at (2.5,-0.9) {013};

\draw (b1.south)--(b2.north);
\draw (b2.south)--(b3.north);

\end{tikzpicture}

\vspace{0.5mm}

$\sigma(0)<\sigma(4)$
\end{minipage}
\hfill
\begin{minipage}[t]{0.29\textwidth}
\centering
\textbf{G14}

\begin{tikzpicture}[scale=0.55, every node/.style={font=\tiny}]

\node[vertex,label=above:0] (v0) at (0,1.4) {};
\node[vertex,label=right:1] (v1) at (1.2,0.5) {};
\node[vertex,label=right:2] (v2) at (0.75,-0.9) {};
\node[vertex,label=left:3] (v3) at (-0.75,-0.9) {};
\node[vertex,label=left:4] (v4) at (-1.2,0.5) {};

\draw[edge] (v0)--(v2);
\draw[edge] (v0)--(v3);
\draw[edge] (v0)--(v4);
\draw[edge] (v1)--(v3);
\draw[edge] (v1)--(v4);
\draw[edge] (v2)--(v4);
\draw[edge] (v3)--(v4);

\node[bag] (b1) at (2.5,0.9) {134};
\node[bag] (b2) at (2.5,0.0) {034};
\node[bag] (b3) at (2.5,-0.9) {024};

\draw (b1.south)--(b2.north);
\draw (b2.south)--(b3.north);

\end{tikzpicture}

\vspace{0.5mm}

$\sigma(0)<\sigma(3)$
\end{minipage}

\caption{Graphs on five vertices detectable in $\widetilde{O}(nm)$-time, with $(1,1)$-width tree decompositions satisfying the required automorphism-avoidance conditions.}
\label{fig:selected-five-vertex-graphs}
\end{figure}
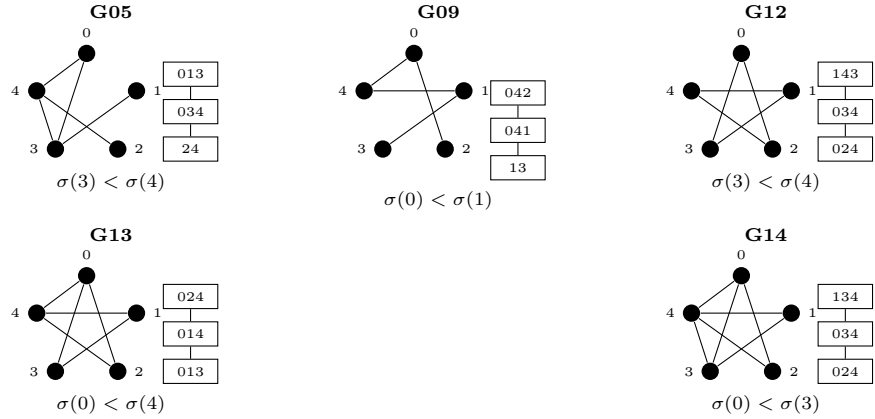

\begin{table}
\caption{Number of occurences of $5$-vertex graphs in $5$-vertex supergraphs}
\setlength{\tabcolsep}{0.7pt}
\begin{NiceTabular}{crrrrrrrrrrrrrrrrrrrrr}
&G01 & G02 & G03 & G04 & G05 & G06 & {G07${}^?$} & {G08${}^?$} & G09 & G10 & G11 & G12 & G13 & G14 & G15
& G16 & G17 & G18
& G19 & G20 & G21 \\
{G01${}^*$} & 1 & 0 & 1 & 0 & 0 & 1 & 0 & 2 & 0 & 0 & 1 & 0 & 0 & 1 & 0 & {$1^*$} & 2 & 0 & 1 & 3 & 5 \\ 
{G02${}^*$} & 0 & 1 & 2 & 2 & 2 & 5 & 6 & 12 & 0 & 1 & 4 & 0 & 4 & 10 & 4 & {$9^*$} & 20 & 10 & 20 & 36 & 60 \\
G03${}^*$ & 0 & 0 & 1 & 0 & 0 & 2 & 0 & 6 & 0 & 0 & 2 & 0 & 0 & 3 & 0 & 3 & 8 & 0 & 4 & 15 & 30 \\
G04${}^*$ & 0 & 0 & 0 & 1 & 0 & 1 & 6 & 6 & 0 & 0 & 0 & 0 & 2 & 4 & 1 & 3 & 12 & 8 & 16 & 30 & 60 \\
G05 & 0 & 0 & 0 & 0 & 1 & 2 & 0 & 6 & 0 & 0 & 0 & 0 & 1 & 5 & 2 & 6 & 14 & 4 & 12 & 30 & 60 \\
G06${}^*$ & 0 & 0 & 0 & 0 & 0 & 1 & 0 & 6 & 0 & 0 & 0 & 0 & 0 & 2 & 0 & 3 & 10 & 0 & 4 & 24 & 60 \\
{G07${}^?$} & 0 & 0 & 0 & 0 & 0 & 0 & 1 & 1 & 0 & 0 & 0 & 0 & 0 & 0 & 0 & 0 & 1 & 1 & 2 & 4 & 10 \\
{G08${}^?$} & 0 & 0 & 0 & 0 & 0 & 0 & 0 & 1 & 0 & 0 & 0 & 0 & 0 & 0 & 0 & 0 & 1 & 0 & 0 & 3 & 10 \\
G09 & 0 & 0 & 0 & 2 & 1 & 2 & 6 & 6 & 1 & 2 & 4 & 5 & 7 & 10 & 4 & 6 & 18 & 14 & 24 & 36 & 60 \\
G10${}^*$ & 0 & 0 & 0 & 0 & 0 & 0 & 0 & 0 & 0 & 1 & 4 & 0 & 2 & 6 & 2 & 3 & 12 & 6 & 16 & 30 & 60 \\
G11${}^*$ & 0 & 0 & 0 & 0 & 0 & 0 & 0 & 0 & 0 & 0 & 1 & 0 & 0 & 1 & 0 & 0 & 2 & 0 & 2 & 6 & 15 \\ G12 & 0 & 0 & 0 & 0 & 0 & 0 & 0 & 0 & 0 & 0 & 0 & 1 & 1 & 1 & 0 & 0 & 2 & 2 & 4 & 6 & 12 \\ 
G13 & 0 & 0 & 0 & 0 & 0 & 0 & 0 & 0 & 0 & 0 & 0 & 0 & 1 & 2 & 0 & 0 & 6 & 4 & 12 & 24 & 60 \\ 
G14 & 0 & 0 & 0 & 0 & 0 & 0 & 0 & 0 & 0 & 0 & 0 & 0 & 0 & 1 & 0 & 0 & 4 & 0 & 4 & 18 & 60 \\
G15${}^*$ & 0 & 0 & 0 & 0 & 0 & 0 & 0 & 0 & 0 & 0 & 0 & 0 & 0 & 2 & 1 & 3 & 8 & 2 & 8 & 24 & 60 \\
G16${}^{K_4}$ & 0 & 0 & 0 & 0 & 0 & 0 & 0 & 0 & 0 & 0 & 0 & 0 & 0 & 0 & 0 & 1 & 2 & 0 & 0 & 6 & 20 \\
G17${}^{K_4}$ & 0 & 0 & 0 & 0 & 0 & 0 & 0 & 0 & 0 & 0 & 0 & 0 & 0 & 0 & 0 & 0 & 1 & 0 & 0 & 6 & 30 \\
G18${}^{K_4}$& 0 & 0 & 0 & 0 & 0 & 0 & 0 & 0 & 0 & 0 & 0 & 0 & 0 & 0 & 0 & 0 & 1 & 1 & 4 & 9 & 30 \\
G19${}^{K_4}$ & 0 & 0 & 0 & 0 & 0 & 0 & 0 & 0 & 0 & 0 & 0 & 0 & 0 & 0 & 0 & 0 & 0 & 0 & 1 & 3 & 15 \\
G20${}^{K_4}$ & 0 & 0 & 0 & 0 & 0 & 0 & 0 & 0 & 0 & 0 & 0 & 0 & 0 & 0 & 0 & 0 & 0 & 0 & 0 & 1 & 10 \\
G21${}^{K_4}$& 0 & 0 & 0 & 0 & 0 & 0 & 0 & 0 & 0 & 0 & 0 & 0 & 0 & 0 & 0 & 0 & 0 & 0 & 0 & 0 & 1 
\end{NiceTabular}
\label{tab:subg5}
\end{table}
\end{proof}

Using the same ideas as for Theorem~\ref{thm:g5}, we can prove the
following two theorems. Since the number of cases are huge, we specify
them outside the main paper in the Appendix.

\begin{theorem}\label{thm:g6}
  Given an $n$-vertex, $m$-edge host graph $H$ as input, there are
  $\widetilde{O}(m^2)$-time algorithms to detect induced subgraphs for
  the sixty-five $6$-vertex graphs in Table~\ref{tab:g6}.
\end{theorem}
\begin{proof}
  For each six-vertex graph $G$, we enumerate all super-graphs that
  contain $G$ an odd number of times.  A table containing these
  subgraph counts \footnote{The program generating this table is also
    given in the same appendix.} is given in a spreadsheet at
  \url{https://github.com/lekhakanjaani/Review/blob/main/Graphs%20on%206%20vertices%20-%206vertices.csv} provides the corresponding computations for graphs on six vertices, analogous to those for graphs on five vertices available at \ref{sec:5vertex}.
    Finally, we give $(0, 2)$-width tree decompositions for these
    graphs along with automorphism avoiding conditions.  The complete
    table is given in Appendix~\ref{sec:6vertex}.
\end{proof}

\begin{theorem}\label{thm:p7c7}
  Given an $m$-edge host graph $H$ as input, we can detect whether $H$
  has an induced $P_7$ ($C_7$) in $\widetilde{O}(nm^2)$-time.
\end{theorem}
\begin{proof}
  We enumerate all $7$-vertex supergraphs of $P_7$ and $C_7$ that
  contain them an odd number of times and show that all of them have
  $(1, 2)$-width tree decompositions and checks that can be
  implemented in those tree decompositions that would make the number
  of remaining automorphisms odd. These can be found in
  Appendix~\ref{sec:7vertex}.
\end{proof}

\section{Induced Paths and Cycles in Bipartite Graphs}
\label{sec:bip}

In this section, we consider a simple extension of the dynamic programming algorithms given by  D{\'{\i}}az, Serna and Thilikos \cite{DBLP:journals/tcs/DiazST02} and Komarath, Kumar, Mishra and Sethia \cite{DBLP:journals/jcss/KomarathKMS26} that constructs polynomials different from the homomorphism polynomial. Suppose a bag $B$ in the tree decomposition of the pattern graph $G$ that we use in the dynamic programming algorithm contains vertices $u$ and $v$ such that $\{u, v\}\not\in E(H)$ for the input graph $H$. Then, we can remove monomials that correspond to any homomorphism $\phi$ from $G$ to $H$ where $\{\phi(u), \phi(v)\} \in E(H)$ while building the arithmetic circuit for $\homp_{G, H}$. We use this idea to obtain the following improved algorithms for detecting induced even length paths and cycles in bipartite graphs.

\begin{theorem}
    Given an $m$-edge bipartite host graph $H$ as input, we can detect whether $H$ has an induced $P_{2k}$ in $\widetilde{O}(m^{k-1})$-time.
\end{theorem}
\begin{proof}
    We consider the following $(0, k-1)$-width path decomposition of $P_{2k}$: $B_1 \text{---} B_2 \text{---} B_3$ where $B_1 = \{1, \dotsc, 2k-2\}$, $B_2 = \{2, \dotsc, 2k-1\}$, and $B_3 = \{3,\dotsc,2k\}$. We add conditions to check that the graph induced on images of $\{1,\dotsc,2k-2\}$ does not have any edges other than that in the labeled $P_{2k}$ using the bag $B_1$. Similar checks can be done for the bags $B_2$ and $B_3$ as well. Therefore, the only edges other than the ones in labeled $P_{2k}$ will be $e_1 = \{1,2k\}$, $e_2 = \{1,2k-1\}$, $e_3 = \{2, 2k\}$. Among these $e_2$ and $e_3$ cannot exist if $\{1,\dotsc,2k\}$ is a path as that would form an odd cycle in $H$, which is not possible in a bipartite graph $H$. Therefore, the only supergraph of $P_{2k}$ appearing as monomials would be $C_{2k}$. But a $C_{2k}$ contains a $P_{2k}$ an even number of times. So evaluation over rings of characteristic two will eliminate it. In addition, we use the condition that $\phi(k) < \phi(k+1)$ to remove one of the automorphisms of $P_{2k}$ from the polynomial. Therefore, the multilinear monomials of the constructed polynomial is in one-to-one correspondence with the monomials of $\indp_{P_{2k}, H}$ (and not $\subp_{P_{2k}, H}$). Since this arithemetic circuit can be constructed and evaluated in $\widetilde{O}(m^{k-1})$-time, the theorem follows.
\end{proof}

\begin{theorem}
    Given an $n$-vertex, $m$-edge bipartite host graph $H$ as input, we can detect whether $H$ has an induced $C_{2k}$ in $\widetilde{O}(nm^{k-1})$-time.
\end{theorem}
\begin{proof}
    We consider the following $(1, k-1)$-width path decomposition for $C_{2k}$: $B_1 \text{---} B_2$ where $B_1 = \{1,\dotsc,2k-1\}$ and $B_2 = \{1,3,\dotsc,2k\}$. Then, we add the check that in $B_1$ and $B_2$ edges other than the ones in the labeled $C_{2k}$ are absent. The only remaining edge that could occur is $\{2, 2k\}$. But this edge in $C_{2k}$ would lead to an odd cycle, which is not possible in a bipartite $H$. We also add the check $\phi(1) < \phi(i)$ for all $i \neq 1$. Therefore, only two automorphisms of $C_{2k}$ will survive. The multilinear part of the resulting polynomial is exactly $2\indp_{C_{2k}, H}$. To check that $\indp_{C_{2k,H}}$ is non-zero, we can simply check whether the computed polynomial has a non-zero multilinear monomial modulo $3$.
\end{proof}

\bibliographystyle{splncs04}
\bibliography{ref}

@inproceedings{DBLP:conf/fsttcs/BlaserKS18,
  author       = {Markus Bl{\"{a}}ser and
                  Balagopal Komarath and
                  Karteek Sreenivasaiah},
  editor       = {Sumit Ganguly and
                  Paritosh K. Pandya},
  title        = {Graph Pattern Polynomials},
  booktitle    = {38th {IARCS} Annual Conference on Foundations of Software Technology
                  and Theoretical Computer Science, {FSTTCS} 2018, Ahmedabad, India,
                  December 11-13, 2018},
  series       = {LIPIcs},
  volume       = {122},
  pages        = {18:1--18:13},
  publisher    = {Schloss Dagstuhl - Leibniz-Zentrum f{\"{u}}r Informatik},
  year         = {2018},
  url          = {https://doi.org/10.4230/LIPIcs.FSTTCS.2018.18},
  doi          = {10.4230/LIPICS.FSTTCS.2018.18},
  bibsource    = {dblp computer science bibliography, https://dblp.org}
}

@article{DBLP:journals/dam/HoangKSS13,
  author       = {Ch{\'{\i}}nh T. Ho{\`{a}}ng and
                  Marcin Kaminski and
                  Joe Sawada and
                  R. Sritharan},
  title        = {Finding and listing induced paths and cycles},
  journal      = {Discret. Appl. Math.},
  volume       = {161},
  number       = {4-5},
  pages        = {633--641},
  year         = {2013},
  url          = {https://doi.org/10.1016/j.dam.2012.01.024},
  doi          = {10.1016/J.DAM.2012.01.024},
  bibsource    = {dblp computer science bibliography, https://dblp.org}
}

@inproceedings{DBLP:conf/focs/0002R21,
  author       = {Marco Bressan and
                  Marc Roth},
  title        = {Exact and Approximate Pattern Counting in Degenerate Graphs: New Algorithms,
                  Hardness Results, and Complexity Dichotomies},
  booktitle    = {62nd {IEEE} Annual Symposium on Foundations of Computer Science, {FOCS}
                  2021, Denver, CO, USA, February 7-10, 2022},
  pages        = {276--285},
  publisher    = {{IEEE}},
  year         = {2021},
  url          = {https://doi.org/10.1109/FOCS52979.2021.00036},
  doi          = {10.1109/FOCS52979.2021.00036},
  bibsource    = {dblp computer science bibliography, https://dblp.org}
}

@inproceedings{DBLP:conf/focs/DalirrooyfardW22,
  author       = {Mina Dalirrooyfard and
                  Virginia {Vassilevska Williams}},
  title        = {Induced Cycles and Paths Are Harder Than You Think},
  booktitle    = {63rd {IEEE} Annual Symposium on Foundations of Computer Science, {FOCS}
                  2022, Denver, CO, USA, October 31 - November 3, 2022},
  pages        = {531--542},
  publisher    = {{IEEE}},
  year         = {2022},
  url          = {https://doi.org/10.1109/FOCS54457.2022.00057},
  doi          = {10.1109/FOCS54457.2022.00057},
  bibsource    = {dblp computer science bibliography, https://dblp.org}
}

@inproceedings{Neetil1985OnTC,
  title={On the complexity of the subgraph problem},
  author={Jaroslav Ne{\v s}etřil and Svatopluk Poljak},
  year={1985},
  url={https://api.semanticscholar.org/CorpusID:117487859}
}

@article{WILLIAMS2005357,
title = {A new algorithm for optimal 2-constraint satisfaction and its implications},
journal = {Theoretical Computer Science},
volume = {348},
number = {2},
pages = {357-365},
year = {2005},
note = {Automata, Languages and Programming: Algorithms and Complexity (ICALP-A 2004)},
issn = {0304-3975},
doi = {https://doi.org/10.1016/j.tcs.2005.09.023},
url = {https://www.sciencedirect.com/science/article/pii/S0304397505005438},
author = {Ryan Williams}
}

@article{DBLP:journals/jcss/KomarathKMS26,
  author       = {Balagopal Komarath and
                  Anant Kumar and
                  Suchismita Mishra and
                  Aditi Sethia},
  title        = {Finding and counting patterns in sparse graphs},
  journal      = {J. Comput. Syst. Sci.},
  volume       = {157},
  pages        = {103728},
  year         = {2026},
  url          = {https://doi.org/10.1016/j.jcss.2025.103728},
  doi          = {10.1016/J.JCSS.2025.103728},
  bibsource    = {dblp computer science bibliography, https://dblp.org}
}

@article{Koutis2009LIMITSAA,
  title={LIMITS and Applications of Group Algebras for Parameterized Problems},
  author={Ioannis Koutis and Ryan Williams},
  journal={ACM Transactions on Algorithms (TALG)},
  year={2009},
  volume={12},
  pages={1 - 18},
  url={https://api.semanticscholar.org/CorpusID:12692062}
}

@article{DBLP:journals/tcs/DiazST02,
  author       = {Josep D{\'{\i}}az and
                  Maria J. Serna and
                  Dimitrios M. Thilikos},
  title        = {Counting H-colorings of partial k-trees},
  journal      = {Theor. Comput. Sci.},
  volume       = {281},
  number       = {1-2},
  pages        = {291--309},
  year         = {2002},
  url          = {https://doi.org/10.1016/S0304-3975(02)00017-8},
  doi          = {10.1016/S0304-3975(02)00017-8},
  bibsource    = {dblp computer science bibliography, https://dblp.org}
}

@inproceedings{DBLP:conf/soda/WilliamsWWY15,
  author       = {Virginia {Vassilevska Williams} and
                  Joshua R. Wang and
                  Richard Ryan Williams and
                  Huacheng Yu},
  editor       = {Piotr Indyk},
  title        = {Finding Four-Node Subgraphs in Triangle Time},
  booktitle    = {Proceedings of the Twenty-Sixth Annual {ACM-SIAM} Symposium on Discrete
                  Algorithms, {SODA} 2015, San Diego, CA, USA, January 4-6, 2015},
  pages        = {1671--1680},
  publisher    = {{SIAM}},
  year         = {2015},
  url          = {https://doi.org/10.1137/1.9781611973730.111},
  doi          = {10.1137/1.9781611973730.111},
  bibsource    = {dblp computer science bibliography, https://dblp.org}
}

@article{doi:10.1137/140978211,
author = {Floderus, Peter and Kowaluk, Miros{\L}aw and Lingas, Andrzej and Lundell, Eva-Marta},
title = {Detecting and Counting Small Pattern Graphs},
journal = {SIAM Journal on Discrete Mathematics},
volume = {29},
number = {3},
pages = {1322-1339},
year = {2015},
doi = {10.1137/140978211}
}

\appendix
\section{An example finding 4-paths in Bull graph $\mathcal{B}$}
For the homomorphism polynomial, $\homp_{P_4,\mathcal{B}}({(y_v)}_{v\in V(\mathcal{B})})$, consider variables $y_v \in R$, where $R$ is a commutative ring with characteristic $\mathrm{char}(R) = 2$, and $y^2 = 0 $ for all $ y \in R$. The automorphism group of $P_4$ has two elements, the identity and the reversal $1\text{-}2\text{-}3\text{-}4 \leftrightarrow 4\text{-}3\text{-}2\text{-}1$. Therefore, each homomorphism from $P_4$ to $\mathcal{B}$ contributes a coefficient of $2$ to the corresponding monomial in the homomorphism polynomial. Since we are working over a ring of characteristic $2$, all coefficients are even and hence vanish identically over $R$.
\begin{figure}[ht]
\centering

\begin{minipage}{0.42\textwidth}
\centering
\begin{tikzpicture}[
    every node/.style={circle,draw,fill=black,inner sep=1.6pt}
]

\node[label=below:$1$] (v1) at (0,0) {};
\node[label=below:$2$] (v2) at (1,0) {};
\node[label=below:$3$] (v3) at (2,0) {};
\node[label=below:$4$] (v4) at (3,0) {};

\draw (v1)--(v2)--(v3)--(v4);

\end{tikzpicture}

\vspace{2mm}

$P_4$
\end{minipage}
\hfill
\begin{minipage}{0.42\textwidth}
\centering

\begin{tikzpicture}[
    every node/.style={circle,draw,fill=black,inner sep=1.6pt}
]

\node[label=left:$c$] (vc) at (0,0) {};
\node[label=left:$b$] (vb) at (0,1) {};
\node[label=above:$a$] (va) at (1,2) {};
\node[label=right:$d$] (vd) at (2,0) {};
\node[label=right:$e$] (ve) at (2,1) {};

\draw (va)--(vb)--(vc);
\draw (va)--(ve)--(vd);
\draw (vb)--(ve);

\end{tikzpicture}

\vspace{2mm}

Bull graph
\end{minipage}

\end{figure}
The even number of automorphisms of $P_4$ can be eliminated by imposing the condition $\phi(2) < \phi(3)$ on the homomorphisms $\phi: V(P_4) \mapsto V(\mathcal{B})$. The four types of monomials present in the homomorphism polynomial are as follows:
\paragraph{Homomorphisms onto a single edge.}
If all three edges of $P_4$ (i.e., $1\text{-}2$, $2\text{-}3$, $3\text{-}4$) map onto the same edge $(i,j) \in E(\mathcal{B})$, the contribution is
\[
\sum_{(i,j)\in E(\mathcal{B})} 2\, y_i^2 y_j^2.
\]
Since this term carries a square factor $y_i^2$, it vanishes identically over $R$.

\paragraph{Homomorphisms onto a 2-edge walk.}
If the image of $P_4$ revisits a vertex, for example, $1\text{-}2\text{-}3\text{-}4 \mapsto a\text{-}b\text{-}c\text{-}a$, the resulting monomial still contains a repeated (squared) variable, and hence again vanishes due to the property of the ring R.

\paragraph{Homomorphisms where $P_4$ embeds as a subgraph.}
When the image of $P_4$ is an actual (injective) subgraph of $\mathcal{B}$, no variable is repeated, and the resulting monomial is multilinear, these are the terms in the homomorphism polynomial which can contribute to induced subgraph counting.

For example, take $1\text{-}2\text{-}3\text{-}4 \mapsto a\text{-}b\text{-}e\text{-}d$. The automorphism group of $P_4$ (reversal, $1\text{-}2\text{-}3\text{-}4 \leftrightarrow 4\text{-}3\text{-}2\text{-}1$) maps this to the same image, contributing the term $y_a y_b y_e y_d$ with coefficient $2$. Because $\mathcal{B}$ contains the extra edge $b\text{-}e$, there is a second homomorphism, $1\text{-}2\text{-}3\text{-}4 \mapsto b\text{-}a\text{-}e\text{-}d$, which (together with its own reversal) contributes the same monomial $y_a y_b y_e y_d$ with coefficient $2$ as well. Combining both, gives a total coefficient $4$ for this monomial. The automorphism condition $\phi(2) < \phi(3)$ eliminates the second homomorphism, leaving a total coefficient $2$ for this monomial. However, the coefficient is still even, and hence this term vanishes identically over $R$.

\paragraph{Homomorphisms where $P_4$ embeds as an induced subgraph.}
Here the image $P_4$ is an \emph{induced} copy in $\mathcal{B}$; with no variable repeated and no extra edge between the vertices in the host graph, resulting in a multilinear monomial with co-efficient 2 (the path and its reversal).

For example, take $1\text{-}2\text{-}3\text{-}4 \mapsto c\text{-}b\text{-}e\text{-}d$. The automorphism group of $P_4$ consists of the identity and the reversal $1\text{-}2\text{-}3\text{-}4 \leftrightarrow 4\text{-}3\text{-}2\text{-}1$; the reversal sends this mapping to $1\text{-}2\text{-}3\text{-}4 \mapsto d\text{-}e\text{-}b\text{-}c$, which is the same walk traversed backward and hence produces the identical monomial $y_c y_b y_e y_d$. 
So when we impose the conditions ($\phi(2) < \phi(3)$) to avoid automorphisms, the coefficient becomes odd for all induced copies.

Hence only the multilinear terms of the homomorphism polynomial contributing to the \emph{induced} subgraph survives, as the terms in the first 3 cases vanish. The terms in the third case vanishes due to even coefficients, since the imposed condition doesn't stop it from counting the extra $P_4$ due to the extra edge.

\label{sec:p4inbull}

\section{5-Vertex Graphs}
\label{sec:5vertex}
Here each of the graph is represented as a 10 bit binary string where each bit represent which of the edges[(0,1), (0,2), (0,3), (0,4),
    (1,2), (1,3), (1,4),
    (2,3), (2,4),
    (3,4)] exists in the graph.

    \vspace{1pt}
\textbf{Good Tree Decomposition}: A tree decomposition is Good for a 5-vertex graph, if the tree width is 2 and the vertices can be spanned by a vertex and an edge.
\vspace{1pt}
\textbf{Conditions}: The conditions helps us to reduce the automorphism to odd automorphisms which helps us in cancelling the even terms present in the equation.

For Example for G02, 0011001010, This graph has two automorphisms as we saw in \ref{sec:p4inbull}, corresponding to swapping vertices 1 and 2. To eliminate this symmetry, we impose the condition $1<2$. We then consider only those homomorphisms that respect this ordering and discard any homomorphism that violates it.

We can produce all the possible graphs on 5-vertices and 6-vertices using the code below
\section{Computational Enumeration of Connected Graphs}
\label{app:connected-graphs-five}

Here we have the Sage/Python code which is used to enumerate connected
$n$-vertex graphs, compute their normalized embedding matrix, encode their
edge sets, test for $K_{n-1}$ minors, and export the matrix to a PDF.
The same implementation is used for the cases $n=5$ and $n=6$.

\subsection{Graph enumeration and edge encoding}
\label{app:graph-enumeration}

The following function enumerates all connected graphs on $n$ vertices,
orders them by automorphism-group size, and assigns a fixed ordering to
their possible edges.  The edge-bitstring gives a compact representation
of each graph with respect to this ordering.

\begin{lstlisting}[language=Python, basicstyle=\ttfamily\scriptsize,
                   breaklines=true, frame=single,
                   caption={Enumeration, ordering, and edge-bitstring encoding
                   of connected $n$-vertex graphs.},
                   label={lst:graph-enumeration}]
from sage.graphs.graph_generators import graphs

def connected_graphs(n):
    G = [H for H in graphs(n) if H.is_connected()]

    labels = [f"G{i+1:03d}" for i in range(len(G))]
    aut = [
        int(H.automorphism_group().cardinality())
        for H in G
    ]

    order = sorted(
        range(len(G)),
        key=lambda i: aut[i]
    )

    G = [G[i] for i in order]
    labels = [labels[i] for i in order]
    aut = [aut[i] for i in order]

    return G, labels, aut


def edge_order(n):
    return [
        (u, v)
        for u in range(n)
        for v in range(u + 1, n)
    ]


def edge_bitstring(G):
    return ''.join(
        '1' if G.has_edge(u, v) else '0'
        for u, v in edge_order(G.order())
    )
\end{lstlisting}

\subsection{Embedding matrix}
\label{app:embedding-matrix}

For graphs $P$ and $H$, the following function counts injective vertex maps
from $P$ to $H$ that preserve every edge of $P$.  The matrix construction
then normalizes this count by the size of the automorphism group of the
pattern graph.

\begin{lstlisting}[language=Python, basicstyle=\ttfamily\scriptsize,
                   breaklines=true, frame=single,
                   caption={Construction of the normalized embedding matrix.},
                   label={lst:embedding-matrix}]
from itertools import permutations

def embedding_count(P, H):
    k = P.order()
    P_vertices = list(P.vertices())
    H_vertices = list(H.vertices())

    P_edges = {
        (min(u, v), max(u, v))
        for u, v, _ in P.edges()
    }

    count = 0

    for S in permutations(H_vertices, k):
        mapping = dict(zip(P_vertices, S))

        if all(
            H.has_edge(mapping[u], mapping[v])
            for u, v in P_edges
        ):
            count += 1

    return count


def embedding_matrix(G, aut):
    m = len(G)
    A = [[0] * m for _ in range(m)]

    for i in range(m):
        for j in range(m):
            A[i][j] = (
                embedding_count(G[i], G[j])
                // aut[i]
            )

    return A
\end{lstlisting}

\subsection{$K_{n-1}$-minor testing}
\label{app:minor-testing}

For an $n$-vertex graph $G$, the following function tests whether $G$ or
its complement contains a $K_{n-1}$ minor.

\begin{lstlisting}[language=Python, basicstyle=\ttfamily\scriptsize,
                   breaklines=true, frame=single,
                   caption={Testing for a $K_{n-1}$ minor in a graph and its complement.},
                   label={lst:minor-testing}]
def has_Kn_minus_1_minor(G):
    n = G.order()
    K = graphs.CompleteGraph(n - 1)

    try:
        has_G_minor = bool(G.minor(K))
    except ValueError:
        has_G_minor = False

    try:
        has_complement_minor = bool(
            G.complement().minor(K)
        )
    except ValueError:
        has_complement_minor = False

    return has_G_minor, has_complement_minor


def check_Kn_minus_1_minors(G):
    return [
        has_Kn_minus_1_minor(H)
        for H in G
    ]
\end{lstlisting}

\subsection{Matrix export and PDF generation}
\label{app:matrix-pdf}

The following function exports the matrix to CSV and generates a PDF
containing only the matrix.  No graph gallery or graph images are generated.

\begin{lstlisting}[language=Python, basicstyle=\ttfamily\scriptsize,
                   breaklines=true, frame=single,
                   caption={Export and PDF generation for the embedding matrix.},
                   label={lst:matrix-pdf}]
import csv
from reportlab.pdfgen import canvas
from reportlab.lib.pagesizes import letter

def export_matrix(A, filename="matrix.csv"):
    with open(filename, "w", newline="") as f:
        writer = csv.writer(f)

        try:
            writer.writerows(A.rows())
        except AttributeError:
            writer.writerows(A)


def save_matrix_pdf(A, labels, filename):
    c = canvas.Canvas(filename, pagesize=letter)
    pw, ph = letter

    c.setFont("Courier", 6)

    x0 = 30
    y = ph - 40
    row_h = 10

    header = "     " + " ".join(
        f"{label:>5}" for label in labels
    )
    c.drawString(x0, y, header)
    y -= 20

    for i, row in enumerate(A):
        line = (
            f"{labels[i]:<5}" +
            " ".join(f"{x:5d}" for x in row)
        )

        c.drawString(x0, y, line)
        y -= row_h

        if y < 40:
            c.showPage()
            c.setFont("Courier", 6)
            y = ph - 40

    c.save()
\end{lstlisting}

\subsection{Input examples}
\label{app:input-examples}

The general framework is instantiated for $n=5$ and $n=6$ by changing only
the input value of $n$.  The following example generates the corresponding
graph family, embedding matrix, matrix PDF, and $K_{n-1}$-minor checks.

\begin{lstlisting}[language=Python, basicstyle=\ttfamily\scriptsize,
                   breaklines=true, frame=single,
                   caption={Example inputs for the five- and six-vertex cases.},
                   label={lst:input-examples}]
for n in [5, 6]:
    G, labels, aut = connected_graphs(n)

    A = embedding_matrix(G, aut)

    export_matrix(
        A,
        filename=f"matrix_{n}.csv"
    )

    save_matrix_pdf(
        A,
        labels,
        filename=f"matrix_{n}.pdf"
    )

    minor_results = check_Kn_minus_1_minors(G)
\end{lstlisting}

\subsection{Computational output}

When n=5, the first computation reports that there are $21$ connected graphs on five
vertices and writes the matrix to the file
\texttt{matrix\_5.pdf}. The second computation outputs the
bitstrings listed in Table~\ref{tab:five-vertex-graphs}.

We can verify the tree decomposition and the conditions for each of the 5-vertex graphs using the code below
\section{Code to verify tree decompositions and automorphism-avoidance conditions}
\label{app:verify-decomposition}
\begin{lstlisting}
from itertools import combinations

def vertices_covered_by_q_edges(H, q):
    """
    Return True if there exist q edges in H
    whose endpoints cover all vertices of H.
    """
    V = set(H.vertices())
    E = list(H.edges(labels=False))

    if not V:
        return True  # empty graph

    if len(V) > 2*q:
        return False  # impossible even in best case

    for edgeset in combinations(E, q):
        covered = set()
        for u, v in edgeset:
            covered.add(u)
            covered.add(v)
        if covered == V:
            return True

    return False


def bag_is_spannable(G, bag, p, q):
    """
    Check if bag can be spanned by p vertices and q edges.
    """
    H = G.subgraph(bag)
    B = list(bag)

    if len(B) < p:
        return False

    for P in combinations(B, p):
        R = [v for v in B if v not in P]
        HR = H.subgraph(R)
        if vertices_covered_by_q_edges(HR, q):
            return True

    return False


def verify_decomposition(G, T, conditions, p, q):
    p = int(p)
    q = int(q)

    # Parse bags
    bags = T.split("-")
    bag_lists = [[int(v) for v in bag] for bag in bags]

    # Parse conditions
    conds = []
    for c in conditions.split(","):
        c = c.strip()
        if not c:
            continue
        i, j = map(int, c.split("<"))
        if i >= j:
            return False, f"Invalid condition {i}<{j}"
        conds.append((i, j))

    # Spanning condition
    for idx, bag in enumerate(bag_lists):
        if not bag_is_spannable(G, bag, p, q):
            return (
                False,
                f"Bag {idx+1} not spannable by p={p}, q={q}"
            )

    # Comparability condition
    for (i, j) in conds:
        found = False
        for bag in bag_lists:
            if i in bag and j in bag:
                found = True
                break
        if not found:
            return False, f"Condition {i}<{j} not witnessed"

    # Automorphism parity condition
    count = 0
    conds = [(str(i), str(j)) for (i, j) in conds]

    for sigma in G.automorphism_group():
        ok = True
        for (i, j) in conds:
            if sigma(i) >= sigma(j):
                ok = False
                break
        if ok:
            count += 1

    if count % 2 == 0:
        return False, f"Automorphism count even ({count})"

    return True, "OK"

G = Graph([('4','3'), ('3','0'), ('0','4'), ('4','2'), ('4','1')])

T = "043-124"
conditions = "1<2, 0<3"

verify_decomposition(G, T, conditions, p=1, q=1)
\end{lstlisting}
\paragraph{Verification procedure.}
The function \texttt{verify\_decomposition} checks whether a given sequence of
bags is valid for parameters \(p\) and \(q\). Each bag must be
\((p,q)\)-spannable: after removing \(p\) vertices, the remaining vertices are
covered by the endpoints of \(q\) edges. This is verified by enumerating all
choices of removed vertices and testing edge coverage in the induced subgraph.

The procedure also checks that every comparability condition \(i<j\) is witnessed by some bag, and that the number of graph automorphisms preserving all such conditions is odd. Any violation causes the verification to fail with an explicit error.

Below is the table of all the possible 5 vertex simple connected graphs with their (p,q) tree decomposition, we can use the above code to verify each tree decomposition is a valid tree decomposition, can be spanned by at most p vertices and q edges, and after applying conditions only odd automorphisms survive. \\
\noindent

\begin{longtable}{ccrllr}
    \label{tab:five-vertex-graphs}
\textbf{Name} & \textbf{Adjacency} & \textbf{\#Aut} &
\textbf{t.d.} & \textbf{Checks} & \textbf{\#AutRem} \\
\hline
\endfirsthead
\textbf{Name} & \textbf{Adjacency} & \textbf{\#Aut} &
\textbf{t.d.} & \textbf{Checks} & \textbf{\#AutRem} \\
\hline
\endhead
G01 & 0001001011 & 24 &124-024-034  & - & - \\
G02 & 0011001010 &2 &034-142   & $1<2$ & 1 \\
G03 & 0011001011 & 4&034-124-234  & $1<2,2<3,3<0$ & 1 \\
G04 & 0011011010 & 2&013-014-24   & $0<1$ & 1\\
G05 & 0011010011 & 2&013-034-24   & $3<4$ & 1 \\
G06 & 0011011011 &2 &034-134-24   & $0<1$ & 1 \\
G07 & 0011011110 & 12 &034-134-234  & - & - \\
G08 & 0011011111 & 12 &034-134-234  & - & - \\
G09 & 0101011000 & 2&042-041-13 & $0<1$ & 1 \\
G10 & 0101011010 & 2&042-134  & $0<2$ & 1 \\
G11 & 0101011011 & 8&024-034-134  & $2<0,0<3,3<1$ & 1 \\
G12 & 0110011010 & 10 &143-034-024 & $3<4$ & 5 \\
G13 & 0111011010 & 2 &024-014-013 & $0<4$ & 1 \\
G14 & 0111011011 & 2 &134-034-024 & $0<3$ & 1 \\
G15 & 0111001110 & 2 &14-024-023 & $0<2$ & 1 \\
G16 & 0111001111 & 6 &- & - & - \\
G17 & 0111011111 & 4 &- & - & - \\
G18 & 0111111010 & 4 &- & - & - \\
G19 & 0111111011 & 8 &- & - & - \\
G20 & 0111111111 & 12 &- & - & - \\
G21 & 1111111111 & 120 &- & - & - \\

\caption{$(1, 1)$-width tree decompositions for $5$-vertex graphs}
\end{longtable}
{\scriptsize

\setlength{\tabcolsep}{2pt}

\begin{NiceTabular}{c|rrrrrrrrrrrrrrrrrrrrr}

[

  No of $G_i$ in $G_j$

]

&G01 & G02 & G03 & G04 & G05 & G06 & \columncolor{c3}{G07} & \columncolor{c3}{G08} & G09 & G10 & G11 & G12 & G13 & G14 & G15

& \columncolor{c1}G16 & \columncolor{c1}G17 & \columncolor{c1}G18

& \columncolor{c1}G19 & \columncolor{c1}G20 & \columncolor{c1}G21 \\

\Block[fill=c2]{1-1}{G01} & 1 & 0 & 1 & 0 & 0 & 1 & 0 & 2 & 0 & 0 & 1 & 0 & 0 & 1 & 0 & \Block[fill=c2]{1-1}{1} & 2 & 0 & 1 & 3 & 5 \\ 

\Block[fill=c2]{1-1}{G02} & 0 & 1 & 2 & 2 & 2 & 5 & 6 & 12 & 0 & 1 & 4 & 0 & 4 & 10 & 4 & \Block[fill=c2]{1-1}{9} & 20 & 10 & 20 & 36 & 60 \\ \rowcolor{c2} G03 & 0 & 0 & 1 & 0 & 0 & 2 & 0 & 6 & 0 & 0 & 2 & 0 & 0 & 3 & 0 & 3 & 8 & 0 & 4 & 15 & 30 \\ \rowcolor{c2} G04 & 0 & 0 & 0 & 1 & 0 & 1 & 6 & 6 & 0 & 0 & 0 & 0 & 2 & 4 & 1 & 3 & 12 & 8 & 16 & 30 & 60 \\ G05 & 0 & 0 & 0 & 0 & 1 & 2 & 0 & 6 & 0 & 0 & 0 & 0 & 1 & 5 & 2 & 6 & 14 & 4 & 12 & 30 & 60 \\ \rowcolor{c2} G06 & 0 & 0 & 0 & 0 & 0 & 1 & 0 & 6 & 0 & 0 & 0 & 0 & 0 & 2 & 0 & 3 & 10 & 0 & 4 & 24 & 60 \\ \Block[fill=c3]{1-1}{G07} & 0 & 0 & 0 & 0 & 0 & 0 & 1 & 1 & 0 & 0 & 0 & 0 & 0 & 0 & 0 & 0 & 1 & 1 & 2 & 4 & 10 \\ \Block[fill=c3]{1-1}{G08} & 0 & 0 & 0 & 0 & 0 & 0 & 0 & 1 & 0 & 0 & 0 & 0 & 0 & 0 & 0 & 0 & 1 & 0 & 0 & 3 & 10 \\ G09 & 0 & 0 & 0 & 2 & 1 & 2 & 6 & 6 & 1 & 2 & 4 & 5 & 7 & 10 & 4 & 6 & 18 & 14 & 24 & 36 & 60 \\ \rowcolor{c2} G10 & 0 & 0 & 0 & 0 & 0 & 0 & 0 & 0 & 0 & 1 & 4 & 0 & 2 & 6 & 2 & 3 & 12 & 6 & 16 & 30 & 60 \\ \rowcolor{c2} G11 & 0 & 0 & 0 & 0 & 0 & 0 & 0 & 0 & 0 & 0 & 1 & 0 & 0 & 1 & 0 & 0 & 2 & 0 & 2 & 6 & 15 \\ G12 & 0 & 0 & 0 & 0 & 0 & 0 & 0 & 0 & 0 & 0 & 0 & 1 & 1 & 1 & 0 & 0 & 2 & 2 & 4 & 6 & 12 \\ 

G13 & 0 & 0 & 0 & 0 & 0 & 0 & 0 & 0 & 0 & 0 & 0 & 0 & 1 & 2 & 0 & 0 & 6 & 4 & 12 & 24 & 60 \\ 

G14 & 0 & 0 & 0 & 0 & 0 & 0 & 0 & 0 & 0 & 0 & 0 & 0 & 0 & 1 & 0 & 0 & 4 & 0 & 4 & 18 & 60 \\ \rowcolor{c2} G15 & 0 & 0 & 0 & 0 & 0 & 0 & 0 & 0 & 0 & 0 & 0 & 0 & 0 & 2 & 1 & 3 & 8 & 2 & 8 & 24 & 60 \\ \rowcolor{c1} G16 & 0 & 0 & 0 & 0 & 0 & 0 & 0 & 0 & 0 & 0 & 0 & 0 & 0 & 0 & 0 & 1 & 2 & 0 & 0 & 6 & 20 \\ \rowcolor{c1} G17 & 0 & 0 & 0 & 0 & 0 & 0 & 0 & 0 & 0 & 0 & 0 & 0 & 0 & 0 & 0 & 0 & 1 & 0 & 0 & 6 & 30 \\ \rowcolor{c1} G18 & 0 & 0 & 0 & 0 & 0 & 0 & 0 & 0 & 0 & 0 & 0 & 0 & 0 & 0 & 0 & 0 & 1 & 1 & 4 & 9 & 30 \\ \rowcolor{c1} G19 & 0 & 0 & 0 & 0 & 0 & 0 & 0 & 0 & 0 & 0 & 0 & 0 & 0 & 0 & 0 & 0 & 0 & 0 & 1 & 3 & 15 \\ \rowcolor{c1} G20 & 0 & 0 & 0 & 0 & 0 & 0 & 0 & 0 & 0 & 0 & 0 & 0 & 0 & 0 & 0 & 0 & 0 & 0 & 0 & 1 & 10 \\ \rowcolor{c1} G21 & 0 & 0 & 0 & 0 & 0 & 0 & 0 & 0 & 0 & 0 & 0 & 0 & 0 & 0 & 0 & 0 & 0 & 0 & 0 & 0 & 1

\end{NiceTabular}
}
Number of occurences of $5$-vertex graphs in $5$-vertex supergraphs
\vspace{0.5cm}
\\
Each of the rows are classified into 3 categories in the above table.\cbox{c2} represents all the graphs which have a valid (p,q) tree decomposition and after applying conditions only odd automorphisms survive. These graphs and graphs with no color assigned are the graphs which we can compute in $\widetilde{O}(n,m)$ time. \cbox{c1} represents those graphs which has a $K_4-minor$. \cbox{c3} represents those graphs which have odd number of occurrences in a $K_4-minor$ graph.

\section{List of 6-Vertex Graphs}
\label{sec:6vertex}
List of all 6-vertex graphs with their corresponding tree
decomposition along with the automorphism-avoidance conditions which
can be detected in $\widetilde{O}(m^2)$-time.
\begin{longtable}{ccrllr}
\textbf{Name} & \textbf{Adjacency} & \textbf{\#Aut} &
\textbf{t.d.} & \textbf{Checks} & \textbf{\#AutRem} \\
\hline
\endfirsthead
\textbf{Name} & \textbf{Adjacency} & \textbf{\#Aut} &
\textbf{t.d.} & \textbf{Checks} & \textbf{\#AutRem} \\
\hline
\endhead
G01 & 001110011001010 & 1 & 0235-0145 & - & 1 \\
G02 & 001110011001001 & 1 & 1245-0345 & - & 1 \\
G03 & 001010011010001 & 1 & 1245-0135 & - & 1 \\
G04 & 001110011010011 & 1 & 1245-0345 & - & 1 \\
G05 & 001100101010011 & 1 & 0245-0135 & - & 1 \\
G06 & 001110101011010 & 1 & 0245-0135 & - & 1 \\
G07 & 001110101010011 & 1 & 0245-0135 & - & 1 \\
G08 & 001110101011110 & 1 & 0135-0234 & - & 1 \\
G09 & 000110010001010 & 2 & 0145-235 & 2$<$3 & 1 \\
G10 & 000110010001011 & 2 & 0145-0235 & 2$<$3 & 1 \\
G11 & 001010011001000 & 2 & 14-0125-03 & 0$<$1 &1  \\
G12 & 001010011001010 & 2 & 0235-415 & 0$<$3 & 1 \\
G13 & 001110011001000 & 2 & 03-0154-25 & 0$<$5 &1  \\
G14 & 001100011001010 & 2 & 1235-0134 & 0$<$4 & 1 \\
G15 & 001110011001011 & 2 & 0345-1245 & 0$<$4 & 1 \\
G16 & 001010011010000 & 2 & 1245-0135 & 1$<$5 & 1 \\
G17 & 001010011011000 & 2 & 035-1245 & 1$<$2 & 1 \\
G18 & 001010011010010 & 2 & 0135-124 & 0$<$3 & 1 \\
G19 & 001010011011001 & 2 & 035-1245 & 1$<$2 & 1 \\
G20 & 001010011010011 & 2 & 0135-1245 & 0$<$3 &1  \\
G21 & 001110011011010 & 2 & 1245-0345 & 1$<$2 & 1 \\
G22 & 001110011011011 & 2 & 0345-1245 & 1$<$2 & 1 \\
G23 & 001110011001110 & 2 & 1245-0345 & 0$<$3 & 1 \\
G24 & 001110011001111 & 2 & 1245-0345 & 0$<$3 & 1 \\
G25 & 001110111001010 & 2 & 1235-0135-014 & 0$<$1 & 1 \\
G26 & 001110100011010 & 2 & 0245-035-13 & 2$<$4 & 1 \\
G27 & 001100101011010 & 2 & 1235-0234 & 3$<$5 & 1 \\
G28 & 001100101011011 & 2 & 0135-0345-0245 & 3$<$4 & 1 \\
G29 & 001110101011011 & 2 & 0135-0345-0245 & 3$<$4 & 1 \\
G30 & 001110111011010 & 2 & 0135-0125-0124 & 0$<$1 & 1 \\
G31 & 001110111011011 & 2 & 1245-0145-0135 & 0$<$1 & 1 \\
G32 & 001110111001110 & 2 & 0134-0135-25 & 0$<$1 & 1 \\
G33 & 001110110001111 & 2 & 1345-0345-25 & 3$<$4 & 1 \\
G34 & 001110101011111 & 2 & 0135-0345-0245 & 3$<$4 & 1 \\
G35 & 010110101011000 & 2 & 0245-135 & 0$<$2 & 1 \\
G36 & 010110111011010 & 2 & 0245-0145-135 & 0$<$2 & 1 \\
G37 & 010110111010011 & 2 & 0245-1345 & 4$<$5 & 1 \\
G38 & 010110111011011 & 2 & 0245-1345 & 0$<$2 & 1 \\
G39 & 011010111011000 & 2 & 1245-0125-0135 & 0$<$2 & 1 \\
G40 & 011010111011010 & 2 & 0235-1235-1245 & 1$<$2 & 1 \\
G41 & 011110111011010 & 2 & 0135-0125-0124 & $1<4$ & 1 \\
G42 & 011110111010011 & 2 & 1345-0345-024 & 1$<$3 & 1 \\
G43 & 011110111011011 & 2 & 0245-0145-0135 & 1$<$3 & 1 \\
G44 & 011110011111010 & 2 & 0235-0245-1245 & 0$<$2 & 1\\
G45 & 011110011111011 & 2 & 0145-0245-0235 & 0$<$2 & 1 \\
G48 & 000110011010010 & 4 & 0124-0145-0135 & 0$<$1, 4$<$5 & 1 \\
G49 & 000110011010011 & 4 & 1245-035 & 0$<$3, 1$<$2 & 1 \\
G54 & 001100011011010 & 4 & 0345-1245 & 0$<$3, 3$<$5 & 1 \\
G56 & 001100011011011 & 4 & 1245-0345 & 0$<$3, 3$<$5 & 1 \\
G57 & 001110010001111 & 4 & 25-0345-14 & 5$<$0, 0$<$4 & 1 \\
G59 & 001110101011000 & 4 & 0135-0245 & 0$<$3, 3$<$1 & 1 \\
G60 & 001110100010011 & 4 & 24-0345-13 & 3$<$5, 5$<$4 & 1 \\
G67 & 010110111010010 & 4 & 0124-0135 & 1$<$4, 4$<$0 & 1 \\
G69 & 010110110011011 & 4 & 0245-4513 & 2$<$0, 0$<$5 & 1 \\
G70 & 011010111010010 & 4 & 0124-0135 & 1$<$4, 4$<$2 & 1 \\
G74 & 000110001001010 & 6 & 1235-045 & 1$<$2, 2$<$3, & 1 \\
G75 & 000110011011010 & 6 & 0124-0125-35 & 0$<$1, 1$<$2 & 1 \\
G76 & 000110011011011 & 6 & 0245-1245-35 & 0$<$2, 2$<$1 & 3 \\
G77 & 001110010001001 & 6& 0345-1245 & 0$<$4, 4$<$5 & 1 \\
G78 & 001100101011111 & 6 & 2345-0345-1345 & 3$<$4,4$<$5 & 1  \\
G79 & 010110101011001 & 6 & 0245-135 & 0$<$2, 2$<$4 & 1 \\
G80 & 011110001111110 & 6 & 0234-0235-15 & 0$<$2, 2$<$3 & 1 \\
G90 & 011010111011011 & 10 & 0135-0145-0245 & 0$<$2 & 5 \\
G93 & 001100101011000 & 12 & 125-124-140-130 & 1$<$5, 5$<$2, 2$<$4, 4$<$0, 0$<$5 & 1 \\
G97 & 011100111011010 & 12 & 0234-1234-1235 & 0$<$2, 2$<$4, 4$<$1 &  1\\
\caption{Six-vertex pattern graphs with $\widetilde{O}(m^2)$-time algorithms for induced subgraph detection.}
\label{tab:g6}
\end{longtable}

\newcommand{\SixGraph}[4]{%
\begin{minipage}[t]{0.31\textwidth}
\centering
\textbf{#1}\par\vspace{1mm}
\begin{tikzpicture}[
    x=0.43cm,y=0.43cm,
    vertex/.style={circle,fill=black,inner sep=1.8pt},
    bag/.style={
        draw,
        inner sep=1.5pt,
        font=\scriptsize,
        minimum height=0.34cm
    },
    every node/.style={font=\scriptsize}
]
\node[vertex,label=above:0] (v0) at (0,1.4) {};
\node[vertex,label=right:1] (v1) at (1.05,0.7) {};
\node[vertex,label=right:2] (v2) at (1.05,-0.7) {};
\node[vertex,label=below:3] (v3) at (0,-1.4) {};
\node[vertex,label=left:4] (v4) at (-1.05,-0.7) {};
\node[vertex,label=left:5] (v5) at (-1.05,0.7) {};

\foreach \u/\v in {#2}{%
    \draw (v\u)--(v\v);%
}

\def\prev{0}%
\foreach \bag [count=\i] in {#3}{%
    \ifcase\i
    \or\def\by{1.00}
    \or\def\by{0.25}
    \or\def\by{-0.50}
    \or\def\by{-1.25}
    \fi
    \node[bag] (b\i) at (2.25,\by) {\bag};
    \ifnum\i>1
        \draw (b\prev.south)--(b\i.north);
    \fi
    \xdef\prev{\i}%
}
\end{tikzpicture}

\if\relax\detokenize{#4}\relax
\else
\par\vspace{-1mm}
$#4$
\fi
\end{minipage}%
}

\SixGraph{G01}{0/3,0/4,0/5,1/4,1/5,2/5,3/5}{0235,0145}{}
\hfill
\SixGraph{G02}{0/3,0/4,0/5,1/4,1/5,2/5,4/5}{1245,0345}{}
\hfill
\SixGraph{G03}{0/3,0/5,1/4,1/5,2/4,4/5}{1245,0135}{}

\par\medskip

\SixGraph{G04}{0/3,0/4,0/5,1/4,1/5,2/4,3/5,4/5}{1245,0345}{}
\hfill
\SixGraph{G05}{0/3,0/4,1/3,1/5,2/4,3/5,4/5}{0245,0135}{}
\hfill
\SixGraph{G06}{0/3,0/4,0/5,1/3,1/5,2/4,2/5,3/5}{0245,0135}{}

\par\medskip

\SixGraph{G07}{0/3,0/4,0/5,1/3,1/5,2/4,3/5,4/5}{0245,0135}{}
\hfill
\SixGraph{G08}{0/3,0/4,0/5,1/3,1/5,2/4,2/5,3/4,3/5}{0135,0234}{}
\hfill
\SixGraph{G09}{0/4,0/5,1/4,2/5,3/5}{0145,235}{2<3}

\par\medskip

\SixGraph{G10}{0/4,0/5,1/4,2/5,3/5,4/5}{0145,0235}{2<3}
\hfill
\SixGraph{G11}{0/3,0/5,1/4,1/5,2/5}{14,0125,03}{0<1}
\hfill
\SixGraph{G12}{0/3,0/5,1/4,1/5,2/5,3/5}{0235,415}{0<3}

\par\medskip

\SixGraph{G13}{0/3,0/4,0/5,1/4,1/5,2/5}{03,0154,25}{0<5}
\hfill
\SixGraph{G14}{0/3,0/4,1/4,1/5,2/5,3/5}{1235,0134}{0<4}
\hfill
\SixGraph{G15}{0/3,0/4,0/5,1/4,1/5,2/5,3/5,4/5}{0345,1245}{0<4}

\par\medskip

\SixGraph{G16}{0/3,0/5,1/4,1/5,2/4}{1245,0135}{1<5}
\hfill
\SixGraph{G17}{0/3,0/5,1/4,1/5,2/4,2/5}{035,1245}{1<2}
\hfill
\SixGraph{G18}{0/3,0/5,1/4,1/5,2/4,3/5}{0135,124}{0<3}

\par\medskip

\SixGraph{G19}{0/3,0/5,1/4,1/5,2/4,2/5,4/5}{035,1245}{1<2}
\hfill
\SixGraph{G20}{0/3,0/5,1/4,1/5,2/4,3/5,4/5}{0135,1245}{0<3}
\hfill
\SixGraph{G21}{0/3,0/4,0/5,1/4,1/5,2/4,2/5,3/5}{1245,0345}{1<2}

\par\medskip

\SixGraph{G22}{0/3,0/4,0/5,1/4,1/5,2/4,2/5,3/5,4/5}{0345,1245}{1<2}
\hfill
\SixGraph{G23}{0/3,0/4,0/5,1/4,1/5,2/5,3/4,3/5}{1245,0345}{0<3}
\hfill
\SixGraph{G24}{0/3,0/4,0/5,1/4,1/5,2/5,3/4,3/5,4/5}{1245,0345}{0<3}

\par\medskip

\SixGraph{G25}{0/3,0/4,0/5,1/3,1/4,1/5,2/5,3/5}{1235,0135,014}{0<1}
\hfill
\SixGraph{G26}{0/3,0/4,0/5,1/3,2/4,2/5,3/5}{0245,035,13}{2<4}
\hfill
\SixGraph{G27}{0/3,0/4,1/3,1/5,2/4,2/5,3/5}{1235,0234}{3<5}

\par\medskip

\SixGraph{G28}{0/3,0/4,1/3,1/5,2/4,2/5,3/5,4/5}{0135,0345,0245}{3<4}
\hfill
\SixGraph{G29}{0/3,0/4,0/5,1/3,1/5,2/4,2/5,3/5,4/5}{0135,0345,0245}{3<4}
\hfill
\SixGraph{G30}{0/3,0/4,0/5,1/3,1/4,1/5,2/4,2/5,3/5}{0135,0125,0124}{0<1}

\par\medskip

\SixGraph{G31}{0/3,0/4,0/5,1/3,1/4,1/5,2/4,2/5,3/5,4/5}{1245,0145,0135}{0<1}
\hfill
\SixGraph{G32}{0/3,0/4,0/5,1/3,1/4,1/5,2/5,3/4,3/5}{0134,0135,25}{0<1}
\hfill
\SixGraph{G33}{0/3,0/4,0/5,1/3,1/4,2/5,3/4,3/5,4/5}{1345,0345,25}{3<4}

\par\medskip

\SixGraph{G34}{0/3,0/4,0/5,1/3,1/5,2/4,2/5,3/4,3/5,4/5}{0135,0345,0245}{3<4}
\hfill
\SixGraph{G35}{0/2,0/4,0/5,1/3,1/5,2/4,2/5}{0245,135}{0<2}
\hfill
\SixGraph{G36}{0/2,0/4,0/5,1/3,1/4,1/5,2/4,2/5,3/5}{0245,0145,135}{0<2}

\par\medskip

\SixGraph{G37}{0/2,0/4,0/5,1/3,1/4,1/5,2/4,3/5,4/5}{0245,1345}{4<5}
\hfill
\SixGraph{G38}{0/2,0/4,0/5,1/3,1/4,1/5,2/4,2/5,3/5,4/5}{0245,1345}{0<2}
\hfill
\SixGraph{G39}{0/2,0/3,0/5,1/3,1/4,1/5,2/4,2/5}{1245,0125,0135}{0<2}

\par\medskip

\SixGraph{G40}{0/2,0/3,0/5,1/3,1/4,1/5,2/4,2/5,3/5}{0235,1235,1245}{1<2}
\hfill
\SixGraph{G41}{0/2,0/3,0/4,0/5,1/3,1/4,1/5,2/4,2/5,3/5}{0135,0125,0124}{1<4}
\hfill
\SixGraph{G42}{0/2,0/3,0/4,0/5,1/3,1/4,1/5,2/4,3/5,4/5}{1345,0345,024}{1<3}

\par\medskip

\SixGraph{G43}{0/2,0/3,0/4,0/5,1/3,1/4,1/5,2/4,2/5,3/5,4/5}{0245,0145,0135}{1<3}
\hfill
\SixGraph{G44}{0/2,0/3,0/4,0/5,1/4,1/5,2/3,2/4,2/5,3/5}{0235,0245,1245}{0<2}
\hfill
\SixGraph{G45}{0/2,0/3,0/4,0/5,1/4,1/5,2/3,2/4,2/5,3/5,4/5}{0145,0245,0235}{0<2}

\par\medskip

\SixGraph{G48}{0/4,0/5,1/4,1/5,2/4,3/5}{0124,0145,0135}{0<1, 4<5}
\hfill
\SixGraph{G49}{0/4,0/5,1/4,1/5,2/4,3/5,4/5}{1245,035}{0<3, 1<2}
\hfill
\SixGraph{G54}{0/3,0/4,1/4,1/5,2/4,2/5,3/5}{0345,1245}{0<3, 3<5}

\par\medskip

\SixGraph{G56}{0/3,0/4,1/4,1/5,2/4,2/5,3/5,4/5}{1245,0345}{0<3, 3<5}
\hfill
\SixGraph{G57}{0/3,0/4,0/5,1/4,2/5,3/4,3/5,4/5}{25,0345,14}{5<0, 0<4}
\hfill
\SixGraph{G59}{0/3,0/4,0/5,1/3,1/5,2/4,2/5}{0135,0245}{0<3, 3<1}

\par\medskip

\SixGraph{G60}{0/3,0/4,0/5,1/3,2/4,3/5,4/5}{24,0345,13}{3<5, 5<4}
\hfill
\SixGraph{G67}{0/2,0/4,0/5,1/3,1/4,1/5,2/4,3/5}{0124,0135}{1<4, 4<0}
\hfill
\SixGraph{G69}{0/2,0/4,0/5,1/3,1/4,2/4,2/5,3/5,4/5}{0245,4513}{2<0, 0<5}

\par\medskip

\SixGraph{G70}{0/2,0/3,0/5,1/3,1/4,1/5,2/4,3/5}{0124,0135}{1<4, 4<2}
\hfill
\SixGraph{G74}{0/4,0/5,1/5,2/5,3/5}{1235,045}{1<2, 2<3,}
\hfill
\SixGraph{G75}{0/4,0/5,1/4,1/5,2/4,2/5,3/5}{0124,0125,35}{0<1, 1<2}

\par\medskip

\SixGraph{G76}{0/4,0/5,1/4,1/5,2/4,2/5,3/5,4/5}{0245,1245,35}{0<2, 2<1}
\hfill
\SixGraph{G77}{0/3,0/4,0/5,1/4,2/5,4/5}{0345,1245}{0<4, 4<5}
\hfill
\SixGraph{G78}{0/3,0/4,1/3,1/5,2/4,2/5,3/4,3/5,4/5}{2345,0345,1345}{3<4,4<5}

\par\medskip

\SixGraph{G79}{0/2,0/4,0/5,1/3,1/5,2/4,2/5,4/5}{0245,135}{0<2, 2<4}
\hfill
\SixGraph{G80}{0/2,0/3,0/4,0/5,1/5,2/3,2/4,2/5,3/4,3/5}{0234,0235,15}{0<2, 2<3}
\hfill
\SixGraph{G90}{0/2,0/3,0/5,1/3,1/4,1/5,2/4,2/5,3/5,4/5}{0135,0145,0245}{0<2}

\par\medskip

\SixGraph{G93}{0/3,0/4,1/3,1/5,2/4,2/5}{125,124,140,130}{1<5, 5<2, 2<4, 4<0, 0<5}
\hfill
\SixGraph{G97}{0/2,0/3,0/4,1/3,1/4,1/5,2/4,2/5,3/5}{0234,1234,1235}{0<2, 2<4, 4<1}
\hfill

List of all 6-vertex graphs for which we have no
$\widetilde{O}(m^2)$-time algorithm for induced subgraph detection.
\begin{longtable}{ccrllr}\label{tab:notm2}
\textbf{Name} & \textbf{Adjacency} & \textbf{\#Aut} &
\textbf{t.d.} & \textbf{Checks} & \textbf{\#AutRem} \\
\hline
\endfirsthead
\textbf{Name} & \textbf{Adjacency} & \textbf{\#Aut} &
\textbf{t.d.} & \textbf{Checks} & \textbf{\#AutRem} \\
\hline
\endhead
G046 & 000110011001010 & 4 & 0145-235 & 0<1, 2<3 & 1 \\

G047 & 000110011001011 & 4 & 0145-235 & 0<1, 2<3 & 1 \\

G050 & 001010010010011 & 4 & 0345-1245 & 0<3, 1<2 & 1 \\

G051 & 001010011011010 & 4 & 1245-035 & 0<3, 1<2 & 1 \\

G052 & 001010011011011 & 4 & 1245-035 & 1<2, 0<3 & 1 \\

G053 & 001110001001110 & 4 & 0345-125 & 1<2, 0<3 & 1 \\

G055 & 001110011011001 & 4 & 1245-0345 & 1<2, 2<5 & 1 \\

G058 & 001110111001000 & 4 & 0134-0125 & 1<4, 4<3 & 1 \\

G061 & 001110111010010 & 4 & 0135-0124 & 5<3, 3<1 & 1 \\

G062 & 001110111011001 & 4 & 1245-0145-013 & 4<5, 5<1 & 1 \\

G063 & 001110111001111 & 4 & 1345-0345-25 & 1<3, 3<0 & 1 \\

G064 & 001110111011111 & 4 & 0345-2345-1235 & 2<1, 1<3 & 1 \\

G065 & 001110111110011 & 4 & 0345-1345-234 & 3<0, 0<4 & 1 \\

G066 & 010110101011010 & 4 & 0245-0125-135 & 2<0, 0<1 & 1 \\

G068 & 010110111011001 & 4 & 0245-145-13 & 4<0, 0<2 & 1 \\

G071 & 011110111011111 & 4 & - & - & - \\

G072 & 011111111011010 & 4 & - & - & - \\

G073 & 011111111011011 & 4 & - & - & - \\

G081 & 000100010001011 & 8 & 014-45-235 & 3<5, 5<4, 4<1 & 1 \\

G082 & 001010011001011 & 8 & 0235-0145 & 3<0, 0<4, 4<1 & 1 \\

G083 & 001110011011110 & 8 & 0345-1245 & 2<5, 5<0, 0<3 & 1 \\

G084 & 001110011011111 & 8 & 0345-1245 & 2<5, 5<0, 0<3 &  \\

G085 & 001110111001011 & 8 & - & - & - \\

G086 & 001110111011000 & 8 & 013-0145-245 & 2<5, 5<1, 1<3, 3<0 &  \\

G087 & 010110101010010 & 8 & 135-045-024 & 3<1, 1<5, 5<4, 4<2 & 1 \\

G088 & 011110011111001 & 8 & 023-0245-145 & 2<3, 3<0, 0<4, 4<1 & 1 \\

G089 & 011111110011011 & 8 & - & - & - \\

G091 & 000110001001011 & 12 & 045-135-125 & 0<4, 1<3 & 3 \\

G092 & 001110001001111 & 12 & 0345-125 & 0<4, 4<3, 1<2 & 1 \\

G094 & 001110111111010 & 12 & - & - & - \\

G095 & 001110111111011 & 12 & - & - & - \\

G096 & 010110101011011 & 12 & 0245-135 & 0<2, 2<4, 1<3 & 1 \\

G098 & 011110011111111 & 12 & - & - & - \\

G099 & 011110111111111 & 12 & - & - & - \\

G100 & 011111111011001 & 12 & - & - & - \\

G101 & 010110111011110 & 16 & 0245-2345-1345 & 0<2, 1<3, 2<3, 4<5 & 1 \\

G102 & 010110111011111 & 16 & 0245-2345-1345 & 0<2, 1<3, 2<3, 4<5 & 1 \\

G103 & 011111111011111 & 16 & - & - & - \\

G104 & 011110001111111 & 24 & - & - & - \\

G105 & 001110111111111 & 36 & - & - & - \\

G106 & 000110011011110 & 48 & 0145-1245-2345 & 0<1, 1<2, 2<3, 4<5 & 1 \\

G107 & 000110011011111 & 48 & 0145-1245-2345 & 0<1, 1<2, 2<3, 4<5 & 1 \\

G108 & 011111111011110 & 48 & - & - & - \\

G109 & 011111111111111 & 48 & - & - & - \\

G110 & 001110111111000 & 72 & - & - & - \\

G111 & 000010001001011 & 120 & - & - & - \\

G112 & 111111111111111 & 720 & - & - & - \\

\caption{List of all 6-vertex graphs with no $\widetilde{O}(m^2)$-time algorithm for induced subgraph detection.}
\end{longtable}

Here the graphs G073, G089,G098-100, G103-104, G109,G111 and G112 are
the graphs which are $K_5$ minor or their complements are $K_5$
minor. For graphs G050, G091 and G094, there are
$(0, 2)$-decompositions but no automorphism avoiding conditions could
be derived for those tree decompositions. All the remaining graphs in
Table~\ref{tab:notm2} are present in one of these graphs an odd number
of times. Therefore, our method cannot detect these in
$\widetilde{O}(m^2)$-time.

\section{7-Vertex Graphs}

We used SageMath to exhaustively enumerate connected graphs on seven
vertices and identify graphs having an odd number of Hamiltonian paths
or Hamiltonian cycles.  For a graph $G$, let $p_7(G)$ denote the number
of undirected Hamiltonian paths and let $c_7(G)$ denote the number of
undirected Hamiltonian cycles.

For the path enumeration, all permutations of the seven vertices were
tested.  A permutation was retained when consecutive vertices were
adjacent, and a path and its reverse were identified as the same
undirected path.  Thus, the retained graphs satisfy
\[
    p_7(G) \equiv 1 \pmod{2}.
\]
For the cycle enumeration, all vertex permutations were similarly
tested, with cycles identified under both rotation and reversal.  Only
connected graphs not already occurring in the $P_7$ collection were
considered in the second enumeration.  The retained graphs therefore
satisfy
\[
    c_7(G) \equiv 1 \pmod{2},
\]
while not belonging to the previously obtained $P_7$ collection.

For each retained graph, the order of its automorphism group was
computed using SageMath.  When required, the corresponding
automorphisms were recorded as vertex maps.  To provide a reproducible
binary representation of each graph, the $21=\binom{7}{2}$ possible
edges were placed in the fixed order
\[
\begin{aligned}
&(0,1),(0,2),(0,3),(0,4),(0,5),(0,6),\\
&(1,2),(1,3),(1,4),(1,5),(1,6),\\
&(2,3),(2,4),(2,5),(2,6),\\
&(3,4),(3,5),(3,6),(4,5),(4,6),(5,6).
\end{aligned}
\]
The corresponding $21$-bit string contains a $1$ exactly when the
associated edge is present and a $0$ otherwise.

The essential enumeration procedure is given below.  The figure
generation and file-format conversion used in the computational
workflow are omitted from the listing, since they do not affect the
graph enumeration itself.

\begin{lstlisting} [language=Python, caption={SageMath code for enumerating $P_7$ and $C_7$ graphs.}]
from itertools import permutations
from sage.graphs.graph_generators import graphs

def hamiltonian_paths(G):
    V = list(G.vertices())
    paths = set()
    for p in permutations(V):
        if all(G.has_edge(p[i], p[i+1])
               for i in range(len(V)-1)):
            paths.add(min(p, p[::-1]))
    return paths

def hamiltonian_cycles(G):
    V = list(G.vertices())
    n = len(V)
    cycles = set()

    for p in permutations(V):
        if not all(G.has_edge(p[i], p[(i+1) % n])
                   for i in range(n)):
            continue

        q = list(p)
        r = q[::-1]
        reps = [tuple(q[i:] + q[:i]) for i in range(n)]
        reps += [tuple(r[i:] + r[:i]) for i in range(n)]
        cycles.add(min(reps))

    return cycles

# Exhaustive enumeration of connected graphs on seven vertices
graphs7 = [G for G in graphs(7) if G.is_connected()]

# Graphs with an odd number of Hamiltonian paths
P7 = []
for G in graphs7:
    p7 = len(hamiltonian_paths(G))
    if p7 % 2 == 1:
        P7.append(
            (G, p7, G.automorphism_group().cardinality())
        )

# Canonical signatures of the P_7 graphs
P7_labels = {
    G.canonical_label().graph6_string()
    for G, _, _ in P7
}

# Graphs with an odd number of Hamiltonian cycles,
# excluding graphs already occurring in P_7
C7 = []
for G in graphs7:
    if G.canonical_label().graph6_string() in P7_labels:
        continue

    c7 = len(hamiltonian_cycles(G))
    if c7 % 2 == 1:
        Aut = G.automorphism_group()
        C7.append(
            (G, c7, Aut.cardinality(), list(Aut))
        )

# Fixed 21-edge encoding
edge_order = [
    (i, j) for i in range(7) for j in range(i + 1, 7)
]

def edge_bitstring(G):
    return ''.join(
        '1' if G.has_edge(u, v) else '0'
        for u, v in edge_order
    )
\end{lstlisting}

The enumeration is exhaustive over the connected graphs generated by
SageMath, while the parity conditions determine the two retained
collections.  The canonical graph labels are used only to exclude
isomorphic copies of graphs already present in the $P_7$ collection
when constructing the $C_7$ collection.

\label{sec:7vertex}
Below is the data for the graphs in which path of length 7 occurs odd number of times.
\begin{longtable}{ccrllr}\label{tab:p7}
\textbf{Name} & \textbf{Adjacency} & \textbf{\#Aut} &
\textbf{t.d.} & \textbf{Checks} & \textbf{\#AutRem} \\
\hline
\endfirsthead
\textbf{Name} & \textbf{Adjacency} & \textbf{\#Aut} &
\textbf{t.d.} & \textbf{Checks} & \textbf{\#AutRem} \\
\hline
\endhead
G001 & 000111001010010001010 & 1 & 01346-025 & - & 1 \\ 
G002 & 000111001000011001010 & 1 & 01346-02456 & - & 1 \\ 
G003 & 000110001010010001011 & 1 & 1346-02456 & - & 1 \\ 
G004 & 000111001010010001011 & 1 & 01346-0256 & - & 1 \\ 
G005 & 000101001000011010010 & 1 & 235-01246 & - & 1 \\ 
G006 & 000111001100011001010 & 1 & 1356-02456 & - & 1 \\ 
G007 & 000111001100011001011 & 1 & 02356-01456 & - & 1 \\ 
G008 & 000111001010011010010 & 1 & 01246-02356 & - & 1 \\ 
G009 & 000111001010011010001 & 1 & 02356-0146 & - & 1 \\ 
G010 & 000111001010011010011 & 1 & 02356-0146 & - & 1 \\ 
G011 & 000111001010010001110 & 1 & 01346-0245 & - & 1 \\ 
G012 & 000111001000011001110 & 1 & 01245-02346 & - & 1 \\ 
G013 & 000111001010010001111 & 1 & 01346-0245 & - & 1 \\ 
G014 & 000111001010011010111 & 1 & 01456-02356 & - & 1 \\ 
G015 & 000111001100101010011 & 1 & 01246-01356 & - & 1 \\ 
G016 & 001011001010011000000 & 1 & 02356-146 & - & 1 \\ 
G017 & 001011001010011000001 & 1 & 12456-02356 & - & 1 \\ 
G018 & 001011001110011001000 & 1 & 02356-1456 & - & 1 \\ 
G019 & 001011001110010001001 & 1 & 1456-02356 & - & 1 \\ 
G020 & 001011001100011001001 & 1 & 145-02356 & - & 1 \\ 
G021 & 001011001100010001011 & 1 & 12456-0356 & - & 1 \\ 
G022 & 001011001110011001001 & 1 & 1456-02356 & - & 1 \\ 
G023 & 001011001110001010010 & 1 & 01356-1246 & - & 1 \\ 
G024 & 001011001110001010011 & 1 & 01356-0246 & - & 1 \\ 
G025 & 001101001110011000000 & 1 & 01346-1256 & - & 1 \\ 
G026 & 001101001110010001000 & 1 & 01346-125 & - & 1 \\ 
G027 & 001101001110010000001 & 1 & 12456-0346 & - & 1 \\ 
G028 & 001101001100011000001 & 1 & 12456-0346 & - & 1 \\ 
G029 & 001101001110011001000 & 1 & 1256-01346 & - & 1 \\ 
G030 & 001101001110011000001 & 1 & 12456-01346 & - & 1 \\ 
G031 & 001101001110010001001 & 1 & 1256-01346 & - & 1 \\ 
G032 & 001101001110011001010 & 1 & 12456-01346 & - & 1 \\ 
G033 & 001111001110011000010 & 1 & 01256-01346 & - & 1 \\ 
G034 & 001111001110010001010 & 1 & 01346-01245 & - & 1 \\ 
G035 & 001111001100011001010 & 1 & 02356-01456 & - & 1 \\ 
G036 & 001111001110011001010 & 1 & 01256-01346 & - & 1 \\ 
G037 & 001111001110011001001 & 1 & 01456-02356 & - & 1 \\ 
G038 & 001111001110011000011 & 1 & 12456-03456 & - & 1 \\ 
G039 & 001111001110001010001 & 1 & 01356-01246 & - & 1 \\ 
G040 & 001111001110001011001 & 1 & 02356-01456 & - & 1 \\ 
G041 & 001101001110011011000 & 1 & 12356-01346 & - & 1 \\ 
G042 & 001101001110010011001 & 1 & 12356-01346 & - & 1 \\ 
G043 & 001101001110011011001 & 1 & 12356-01346 & - & 1 \\ 
G044 & 001101001110011010011 & 1 & 03456-12456 & - & 1 \\ 
G045 & 001110001110001001110 & 1 & 12456-03456 & - & 1 \\ 
G046 & 001111001110001001110 & 1 & 0236-01456 & - & 1 \\ 
G047 & 001111001110010011011 & 1 & 01456-02356 & - & 1 \\ 
G048 & 001111001110011001110 & 1 & 12456-03456 & - & 1 \\ 
G049 & 001111001110011001101 & 1 & 02356-01456 & - & 1 \\ 
G050 & 001110001110011001111 & 1 & 03456-12456 & - & 1 \\ 
G051 & 001011001110110000010 & 1 & 01356-1246 & - & 1 \\ 
G052 & 001011001110110001010 & 1 & 01346-01245 & - & 1 \\ 
G053 & 001011001110101010011 & 1 & 01356-12456 & - & 1 \\ 
G054 & 001011001110101001100 & 1 & 01356-12456 & - & 1 \\ 
G055 & 001011001110100001101 & 1 & 1245-01356 & - & 1 \\ 
G056 & 001011001100101001101 & 1 & 0356-12456 & - & 1 \\ 
G057 & 001011001100100001111 & 1 & 12456-01356 & - & 1 \\ 
G058 & 001011001110101001101 & 1 & 0356-12456 & - & 1 \\ 
G059 & 001011001110100001111 & 1 & 12456-0356 & - & 1 \\ 
G060 & 001011001100101001111 & 1 & 03456-12456 & - & 1 \\ 
G061 & 001010001110101001111 & 1 & 12456-0356 & - & 1 \\ 
G062 & 001011001110110001111 & 1 & 0356-12456 & - & 1 \\ 
G063 & 001111001110101011110 & 1 & 03456-12456 & - & 1 \\ 
G064 & 001111001110101011101 & 1 & 01356-01456-01246 & - & 1 \\ 
G065 & 001111001110011100010 & 1 & 01346-01256 & - & 1 \\ 
G066 & 001111001100011101001 & 1 & 02356-01345 & - & 1 \\ 
G067 & 001110001110010101011 & 1 & 01256-01346 & - & 1 \\ 
G068 & 001111001110011101010 & 1 & 01256-01346 & - & 1 \\ 
G069 & 001111001110011101001 & 1 & 01256-01346 & - & 1 \\ 
G070 & 001111001110011100011 & 1 & 01346-01256 & - & 1 \\ 
G071 & 001111001100011101011 & 1 & 03456-12456 & - & 1 \\ 
G072 & 001101010100111001000 & 1 & 02346-12356 & - & 1 \\ 
G073 & 001111010110110001010 & 1 & 01346-01245 & - & 1 \\ 
G074 & 001111010110111001001 & 1 & 01356-02456 & - & 1 \\ 
G075 & 001111010010111011010 & 1 & 02456-01356 & - & 1 \\ 
G076 & 001101010100111011010 & 1 & 02346-12356 & - & 1 \\ 
G077 & 001100010110111011010 & 1 & 23456-03456-13456 & - & 1 \\ 
G078 & 001101010110111011010 & 1 & 02346-12356 & - & 1 \\ 
G079 & 001101010110110011011 & 1 & 01356-02456 & - & 1 \\ 
G080 & 001111010110111011010 & 1 & 01356-02456 & - & 1 \\ 
G081 & 001111010100111011011 & 1 & 01356-02456 & - & 1 \\
G082 & 001110010110111011011 & 1 & 01356-02456 & - & 1 \\ 
G083 & 001101010110110011111 & 1 & 01356-02456 & - & 1 \\ 
G084 & 001110010110111101001 & 1 & 02345-12356 & - & 1 \\ 
G085 & 001111010110111101001 & 1 & 02346-02356-01356 & - & 1 \\ 
G086 & 001111010110111111010 & 1 & 23456-03456-13456 & - & 1 \\ 
G087 & 010111010110110001011 & 1 & 02456-01356 & - & 1 \\ 
G088 & 010111011110100011010 & 1 & 01246-01356 & - & 1 \\ 
G089 & 010111011110101011010 & 1 & 01356-01246 & - & 1 \\ 
G090 & 010111011110101011110 & 1 & 02456-13456 & - & 1 \\ 
G091 & 000110001000011001010 & 2 & 02456-02346-01246 & 4$<$6 & 1 \\ 
G092 & 000101001000011010011 & 2 & 02356-01246 & 4$<$5 & 1 \\ 
G093 & 000101001000011010111 & 2 & 23456-12456-02456 & 4$<$5 & 1 \\ 
G094 & 001010001010011000000 & 2 & 12456-02346 & 5$<$6 & 1 \\ 
G095 & 001011001010010001000 & 2 & 02356-01346 & 0$<$6 & 1 \\ 
G096 & 001011001110001010000 & 2 & 01356-01456-01256 & 1$<$6 & 1 \\ 
G097 & 001011001110011010010 & 2 & 01456-01256-01356 & 5$<$6 & 1 \\ 
G098 & 001011001110011010011 & 2 & 01456-02356 & 5$<$6 & 1 \\ 
G099 & 001101001100011001000 & 2 & 02356-01356-01346 & 0$<$6 & 1 \\ 
G100 & 001100001110011001000 & 2 & 12356-01356-01456 & 1$<$6 & 1 \\ 
G101 & 001101001100011001001 & 2 & 12456-01456-03456 & 0$<$5 & 1 \\ 
G102 & 001101001110010000011 & 2 & 01346-12456 & 1$<$4 & 1 \\ 
G103 & 001111001110010000011 & 2 & 01256-01456-03456 & 0$<$5 & 1 \\ 
G104 & 001110001110011011010 & 2 & 01436-01456-12456 & 1$<$4 & 1 \\ 
G105 & 001110001110011011011 & 2 & 01456-01356-12356 & 5$<$6 & 1 \\ 
G106 & 001111001100011001111 & 2 & 02456-03456-01456 & 5$<$6 & 1 \\ 
G107 & 001011001110101010000 & 2 & 01356-12456 & 1$<$6 & 1 \\ 
G108 & 001010001100101001111 & 2 & 02346-03456-01345 & 5$<$6 & 1 \\ 
G109 & 001111001100011100011 & 2 & 03456-02456-01456 & 0$<$6 & 1 \\ 
G110 & 001111011100011011011 & 2 & 02356-01356-01346 & 5$<$6 & 1 \\ 
G111 & 001101010110110001010 & 2 & 01346-12456 & 0$<$3 & 1 \\ 
G112 & 001111010110110001001 & 2 & 01356-02456 & 0$<$5 & 1 \\ 
G113 & 001111010110111010011 & 2 & 02456-01356 & 5$<$6 & 1 \\ 
G114 & 001111010100111101001 & 2 & 01356-02356-2346 & 3$<$5 & 1 \\ 
G115 & 001111010110111110011 & 2 & 03456-03456-23456 & 3$<$4 & 1 \\ 
G116 & 010111011110101011000 & 2 & 01356-01246 & 0$<$1 & 1 \\ 
G117 & 010111011110110011010 & 2 & 13456-02456 & 4$<$6 & 1 \\ 
G118 & 011011011110111010011 & 2 & 12456-01256-01356 & 5$<$6 & 1 \\ 
G119 & 011101011110111011000 & 2 & 01236-01234-12345 & 0$<$3 & 1 \\ 
G120 & 011111011110111011010 & 2 & 01246-01236-01235 & 1$<$5 & 1 \\ 
G121 & 001100001100011001000 & 14 & 01346-01236-12356 & 3$<$4 & 7 \\ 
G122 & 011110011110111011010 & 14 & 01246-01236-01235 & 1$<$4 & 7 \\ 
\caption{7-vertex graphs with odd number of occurrences of paths of length 7}
\end{longtable}

\vskip 1 cm
Below is the data for the graphs in which cycles of length 7 occur an odd number of times and did not occur in the above graphs.\par

\begin{longtable}{ccrllr}
\textbf{Name} & \textbf{Adjacency} & \textbf{\#Aut} &
\textbf{t.d.} & \textbf{Checks} & \textbf{\#AutRem} \\
\hline
\endfirsthead
\textbf{Name} & \textbf{Adjacency} & \textbf{\#Aut} &
\textbf{t.d.} & \textbf{Checks} & \textbf{\#AutRem} \\
\hline
\endhead
G01 & 001101001100011001010 & 1 & 0346-12456 & - & 1 \\
G02 & 001101001110011001001 & 1 & 0346-12456 & - & 1 \\
G03 & 001101001110011010010 & 1 & 02356-01456 & - & 1 \\
G04 & 001111001110011001111 & 1 & 01456-02356 & - & 1 \\
G05 & 001011001110110001001 & 1 & 01356-1245 & - & 1 \\
G06 & 001011001110110001011 & 1 & 0356-12456 & - & 1 \\
G07 & 001011001110101010010 & 1 & 01356-1246 & - & 1 \\
G08 & 001010001110101001110 & 1 & 12456-01356 & - & 1 \\
G09 & 001010001110101001101 & 1 & 12456-0356 & - & 1 \\
G10 & 001011001110101001110 & 1 & 12456-0356 & - & 1 \\
G11 & 001011001110101001111 & 1 & 12456-0356 & - & 1 \\
G12 & 001101001110110010011 & 1 & 1245-03456 & - & 1 \\
G13 & 001111001110101011111 & 1 & 12456-03456 & - & 1 \\
G14 & 001110001110011100011 & 1 & 12456-03456 & - & 1 \\
G15 & 001110001100011101011 & 1 & 03456-12456 & - & 1 \\
G16 & 001111001110011101011 & 1 & 01256-01346 & - & 1 \\
G17 & 001111010110111001000 & 1 & 02456-01356 & - & 1 \\
G18 & 001111010010111011001 & 1 & 02456-01356 & - & 1 \\
G19 & 001111010110111011011 & 1 & 02456-01356 & - & 1 \\
G20 & 001111010110101011110 & 1 & 01356-02456 & - & 1 \\
G21 & 001111011010011111101 & 1 & 02356-03456-1346 & - & 1 \\
G22 & 001110010110110101011 & 1 & 12456-01456-01346 & - & 1 \\
G23 & 001111010110110101011 & 1 & 01356-03456-2345 & - & 1 \\
G24 & 001111010110111111011 & 1 & 02346-02356-1356 & - & 1 \\
G25 & 010111011010110010011 & 1 & 02456-13456 & - & 1 \\
G26 & 011111011110111010011 & 1 & 01345-01456-02456 & - & 1 \\
G27 & 001100001110011001010 & 2 & 12456-01346 & 1$<$4 & 1 \\ 
G28 & 001101001110011001011 & 2 & 12456-0346 & 1$<$4 & 1 \\ 
G29 & 001111001110011001000 & 2 & 01346-01256 & 0$<$6 & 1 \\ 
G30 & 001110001110011001010 & 2 & 03465-12456 & 0$<$5 & 1 \\ 
G31 & 001111001110011001011 & 2 & 01456-02356 & 0$<$5 & 1 \\ 
G32 & 001101001110011010000 & 2 & 12356-01346 & 1$<$6 & 1 \\ 
G33 & 001011001100101010010 & 2 & 01356-1246 & 0$<$6 & 1 \\ 
G34 & 001011001100101001110 & 2 & 0145-02346 & 4$<$6 & 1 \\ 
G35 & 001111010100111001001 & 2 & 12356-02356-0246 & 2$<$3 & 1 \\ 
G36 & 001111010100101011111 & 2 & 01356-02456 & 5$<$6 & 1 \\ 
G37 & 001111010110111101000 & 2 & 1356-02356-02346 & 2$<$5 & 1 \\ 
G38 & 010111010110101010001 & 2 & 1356-02456 & 0$<$5 & 1 \\ 
G39 & 010111011010110011010 & 2 & 02345-01346 & 0$<$6 & 1 \\ 
G40 & 010111011110101011111 & 2 & 02456-13456 & 4$<$5 & 1 \\ 
G41 & 011011011110110011011 & 2 & 02356-23456-13456 & 2$<$4 & 1 \\ 
G42 & 011111011110111011011 & 2 & 01356-01256-01246 & 0$<$5 & 1 \\
\caption{7-vertex graphs with odd number of occurrences of cycles of length 7 and not in Table~\ref{tab:p7}}
\end{longtable}

\end{document}